\documentclass[acmsmall,screen,nonacm,natbib=false]{acmart}

\setcopyright{cc}
\setcctype{by}
\acmDOI{10.1145/3763112}
\acmYear{2025}
\acmJournal{PACMPL}
\acmVolume{9}
\acmNumber{OOPSLA2}
\acmArticle{334}
\acmMonth{10}
\received{2025-03-25}
\received[accepted]{2025-08-12}

\RequirePackage[
  datamodel=acmdatamodel,
  style=acmnumeric,
  backref=true,
  backrefstyle=three,
  uniquename=false,
  uniquelist=false
  ]{biblatex}
\usepackage{software-biblatex}

\usepackage{mathtools}
\usepackage{multicol}
\usepackage{bcprules}
\usepackage{commands}
\usepackage{xspace}
\usepackage{fancybox}
\usepackage[shortlabels]{enumitem}
\newcommand{\capless}{\textsf{Capless}\xspace}
\newcommand{\calculus}{\capless}
\newcommand{\ccformal}{\textsf{CC}\ensuremath{_{<:\Box}}\xspace}
\newcommand{\cappy}{\textsf{Reacap}\xspace}
\newcommand{\capcalculus}{\cappy}
\newcommand{\decap}{\textsf{Decap}\xspace}

\newcommand{\smallcode}[1]{{\footnotesize\texttt{#1}}}

\usepackage{mdframed}

\usepackage{listings}

\lstdefinelanguage{dotty}{
  basicstyle=\footnotesize\ttfamily,
  keywords={erased, val, var, if, then, in, handle,
    return, def, match, case, new, type, trait,
     package, object, given, eff,
     pretype, class, extends, extension, infix, else,
     box, unbox, try, catch, import, throw, throws, using,
		 use, cap, box, unbox, extension, this, @use},
  keywordstyle=\bfseries,
  sensitive=true,
  comment=[l]{//},
  morecomment=[s]{/*}{*/},
  commentstyle=\color{ACMDarkBlue},
  stringstyle=\color{gray}, %
  morestring=[b]',
  morestring=[b]",
  moredelim=**[is][\btHL]{`}{`},
	columns=fullflexible,
}
\usepackage{thmtools}
\usepackage{thm-restate}

\usepackage{cleveref}
\usepackage{stmaryrd}

\usepackage{wrapfig}

\newtheorem{definition}{Definition}[section]
\newtheorem{theorem}{Theorem}[section]

\newcommand{\mknote}[3]{{\color{#1}\textbf{\fcolorbox{blue!20}{blue!20}{#2}: {#3}}}}
\renewcommand{\mknote}[3]{}

\usepackage{yfonts,lettrine}

\AtBeginDocument{\setlength{\DefaultFindent}{0.5em}}
\newcommand{\pparagraph}[1]{\paragraph{\normalfont\textbf{\textsf{#1}}}}

\newif\ifextended

\newcommand{\genericappendixref}[3]{%
  \ifextended
    {#3}~\ref{#2}%
  \else
    \cite[{#3}~#1]{whatisinthebox}%
  \fi
}
\newcommand{\appendixref}[2]{%
  \genericappendixref{#1}{#2}{Appendix}%
}
\newcommand{\appendixfigref}[2]{%
  \genericappendixref{#1}{#2}{Figure}%
}

\newcommand{\condtext}[2]{%
  \ifextended
    #1%
  \else
    #2%
  \fi
} \extendedfalse    %
\extendedtrue   %

\makeatletter%
\begin{document}

\title{What's in the Box}
\subtitle{Ergonomic and Expressive Capture Tracking over Generic Data Structures (Extended Version)}

\author{Yichen Xu}
\orcid{0000-0003-2089-6767}
\affiliation{%
  \institution{EPFL}
  \city{Lausanne}
  \country{Switzerland}
}
\email{yichen.xu@epfl.ch}

\author{Oliver Bračevac}
\orcid{0000-0003-3569-4869}
\affiliation{%
  \institution{EPFL}
  \city{Lausanne}
  \country{Switzerland}
}
\email{oliver.bracevac@epfl.ch}

\author{Cao Nguyen Pham}
\orcid{0009-0005-2543-3309}
\affiliation{%
  \institution{EPFL}
  \city{Lausanne}
  \country{Switzerland}
}
\email{nguyen.pham@epfl.ch}

\author{Martin Odersky}
\orcid{0009-0005-3923-8993}
\affiliation{%
  \institution{EPFL}
  \city{Lausanne}
  \country{Switzerland}
}
\email{martin.odersky@epfl.ch}

\begin{abstract}
Capturing types in Scala unify static effect and resource tracking with object capabilities,
enabling lightweight effect polymorphism with minimal notational overhead. However, their
expressiveness has been insufficient for tracking capabilities embedded in generic data structures,
preventing them from scaling to the standard collections library -- an essential prerequisite for
broader adoption. This limitation stems from the inability to name capabilities within the system's
notion of box types.

This paper develops System \capless, a new foundation for capturing types
that provides the theoretical basis for reach capabilities (rcaps), a novel mechanism for naming
``what's in the box''.
The calculus refines the universal capability notion into a new scheme with existential and
universal capture set quantification. Intuitively, rcaps witness existentially quantified capture
sets inside the boxes of generic types in a way that does not require exposing existential capture
types in the surface language. We have fully mechanized the formal metatheory of System \capless{}
in Lean, including proofs of type soundness and scope safety.
System \capless{} supports the same lightweight notation of capturing types plus rcaps, as
certified by a type-preserving translation, and also enables fully optional explicit capture-set
quantification to increase expressiveness.

Finally, we present a full reimplementation of capture checking in Scala 3 based on System Capless
and migrate the entire Scala collections library and an asynchronous programming library to evaluate
its practicality and ergonomics. Our results demonstrate that reach capabilities enable the adoption
of capture checking in production code with minimal changes and minimal-to-zero notational overhead
in a vast majority of cases.
\end{abstract}

\begin{CCSXML}
  <ccs2012>
     <concept>
         <concept_id>10003752.10010124.10010125.10010130</concept_id>
         <concept_desc>Theory of computation~Type structures</concept_desc>
         <concept_significance>500</concept_significance>
         </concept>
     <concept>
         <concept_id>10011007.10011006.10011008.10011024.10011025</concept_id>
         <concept_desc>Software and its engineering~Polymorphism</concept_desc>
         <concept_significance>300</concept_significance>
         </concept>
   </ccs2012>
\end{CCSXML}

\ccsdesc[500]{Theory of computation~Type structures}
\ccsdesc[300]{Software and its engineering~Polymorphism}

\keywords{Scala, Capture Checking, Effect Polymorphism, Generic Data Structures}

\maketitle

\section{Introduction}\label{sec:introduction}

Statically tracking effects and resources through type systems
has attracted increasing research efforts in programming languages \cite{DBLP:conf/popl/LindleyMM17,DBLP:journals/jfp/ConventLMM20,DBLP:journals/pacmpl/BrachthauserSO20,DBLP:journals/pacmpl/BrachthauserSLB22,DBLP:journals/corr/abs-2407-11816}.
Despite this growing body of research,
integrating effect systems into mainstream programming languages
remains challenging due to concerns about usability and flexibility.
Capturing Types (CT)\footnote{We also refer to the approach of Capturing Types as \emph{capture checking} and \emph{capture tracking}.} \cite{DBLP:journals/toplas/BoruchGruszeckiOLLB23} is a promising advancement
that applies the object-capability model \cite{objectcapabilites}
to provide a simple, safe, and practical foundation for effect tracking in Scala.
Developing an effect system for such an established language
brings both opportunities and constraints:
while the existing ecosystem facilitates adoption,
legacy designs and pre-existing code bases
impose strict requirements
for ergonomics and backward compatibility.

\vspace{-5pt}
\pparagraph{Bringing Effect Tracking to the Masses}
The key to making effect tracking practical and usable in established languages is to describe
effect polymorphism without sacrificing flexibility. Since effects are transitive along call edges,
every higher-order function needs to be effect-polymorphic to account for effects performed by its
arguments. In effect systems with explicit quantifiers, effect parameters in signatures are
required. For instance, the familiar \lstinline|map| function would become the following in a
hypothetical effect system with explicit quantifiers:
\begin{lstlisting}[aboveskip=1pt,belowskip=1pt]
class List[+T] { def map[U, E](f: T -> U eff E): List[U] eff E }
\end{lstlisting}
While such verbosity is acceptable in principle, especially when designing a new language, it
quickly becomes disruptive in established languages like Scala where signature complexity grows
rapidly. CT addresses this challenge through lightweight effect signatures and implicit effect
polymorphism, a design goal shared by many recent works on other languages
\cite{DBLP:conf/popl/LindleyMM17,DBLP:journals/jfp/ConventLMM20,DBLP:journals/pacmpl/BrachthauserSO20,DBLP:journals/pacmpl/BrachthauserSLB22,DBLP:journals/corr/abs-2407-11816}.

In CT's approach, effects are performed and resources are accessed via capabilities,
which are objects that are referenced and tracked as regular program variables.
CT's essence is to \emph{track captured variables in types}. It introduces \emph{capturing types}
that augment regular types with a \emph{capture set} over-approximating the variables a value may
capture, and thus giving a handle on the effects a value may perform. Consider a function that
greets someone using console I/O:
\begin{lstlisting}[aboveskip=1pt,belowskip=1pt]
(name: String) => console.println("Hello " + name)
\end{lstlisting}
Assuming \lstinline|console| is a capability for console I/O,
this function has type \lstinline|String ->{console} Unit| (shorthand for \lstinline|(String -> Unit)^{console}|).
This \emph{capturing type} consists of
(1) the \emph{shape type} specifying parameter and return types, and
(2) the \emph{capture set} \lstinline|{console}| indicating the function may capture \lstinline|console|.
The type makes it evident that the function may perform console I/O.

CT's lightweight notation enables effect polymorphism that is minimally invasive.
Returning to our earlier example, the signature of \lstinline|map| remains \emph{unchanged} from vanilla Scala:
\begin{lstlisting}[aboveskip=1pt,belowskip=1pt]
class List[+T] { def map[U](f: T => U): List[U] /* no signature changes needed */ }
\end{lstlisting}
The function arrow \lstinline|T => U| is shorthand for \lstinline|(T -> U)^{cap}|,
where \lstinline|{cap}| is the top capture subsuming all capabilities.
This allows argument \lstinline|f| to capture arbitrary capabilities,
making \lstinline|map| effect-polymorphic, analogous to modeling polymorphism through subtyping and \lstinline|Object| in OO languages.

\vspace{-5pt}
\pparagraph{Retrofitting Capturing Types into Scala}
Another key concern is harmoniously and non-invasively integrating the new universe of capturing
types (like \lstinline|String ->{console} Unit|) into the existing Scala type system. Previous work
(System \ccformal{} \cite{DBLP:journals/toplas/BoruchGruszeckiOLLB23}) formally studied a pragmatic
and sound way of retrofitting this universe in a model $\lambda$-calculus that
informed Scala's CT implementation. Consider a set of concurrent futures
\lstinline|Set[Future[T]^]|, where the hat annotation \lstinline|Future[T]^| (shorthand for
\lstinline|Future[T]^{cap}|) indicates that futures are tracked as capabilities. Under the hood,
\ccformal{} requires generic type arguments with captures to be \emph{boxed}, i.e., this type
desugars to \lstinline|Set[box Future[T]^{cap}]|. The appeal of boxing is that existing
generic types like \lstinline|Set[T]| do not need to be polluted with extra quantifiers for captures,
and work with normal Scala types as well as capturing types,
contributing to the lightweight nature of CT.

\vspace{-5pt}
\pparagraph{Generic Types, Elusive Captures}
Attempts to apply CT to Scala's collection library, however, revealed a key limitation of
\ccformal{} that makes using generic data structures impractical.
Consider a function that turns a set of futures
into a stream arranging them by completion order:
\begin{lstlisting}[aboveskip=1pt,belowskip=1pt]
def collect[T](fs: Set[Future[T]^]): Stream[Future[T]^] =
  val channel = Channel()
  fs.forEach(_.onComplete(v => channel.send(v)))  // error, elements inaccessible!
  Stream.of(channel)
\end{lstlisting}
This function is \emph{untypeable} in \ccformal{}!
To prevent unsafe capability leaks, the calculus forbids accessing a boxed value when its capture set is the top element \lstinline|{cap}| \cite{DBLP:journals/toplas/BoruchGruszeckiOLLB23}.
Hence, the futures in \lstinline|fs| are inaccessible and \lstinline|collect| does not type-check
(we detail boxes and \lstinline|{cap}| in \Cref{sec:motivation}).
\ccformal{} lacks the ability to handle \emph{nested} captures within generic types.

Furthermore, even if the type system tolerated capability leaks and allowed unboxing \lstinline|cap|-qualified types,
the function would be impractical.
The result type \lstinline|Stream[box Future[T]^{cap}]| is imprecise:
it suggests futures in the stream may perform \emph{arbitrary} effects,
despite these futures originating from the input \lstinline|fs|
and performing at most the effects of futures in that collection.
\ccformal{} cannot express this precise input-output relationship.
Is boxing doomed?

\vspace{-5pt}
\pparagraph{The Problem: What's in the Box?}
The fundamental problem is the lack of a mechanism for 
safely accessing and referring to capabilities inside boxes.
To see the core issue,
we approach the problem from the angle of explicit capture quantifications.
The universal capture \lstinline|cap| can be understood as an existential capture.
For instance, the type \lstinline|Future[T]^{cap}| means a future capturing \emph{some} arbitrary capabilities; so it can be viewed as \lstinline[mathescape=true]|$\exists c.$Future[T]^{$c$}|.
Under this perspective, the \lstinline|collect| signature then becomes:
\begin{lstlisting}[mathescape=true]
def collect[T](fs: Set[box $\exists c_1.$Future[T]^{$c_1$}]): Stream[box $\exists c_2.$Future[T]^{$c_2$}]
\end{lstlisting}
This reveals the disconnect between parameter and result captures.
Since result futures originate from input futures,
a more precise and desirable signature would be:
\begin{lstlisting}[mathescape=true]
def collect[T][$\forall c_1$](fs: Set[box Future[T]^{$c_1$}]): Stream[box Future[T]^{$c_1$}]
\end{lstlisting}
Here, the witness $c_1$ flows from input to output.
Unfortunately, this signature is inexpressible in \ccformal{}.
The lesson to be learned here is that ``no two \lstinline|cap|s are created equal'': we need
more granular means to distinguish between them.

As one solution, the CT system we propose supports \emph{optional} explicit capture polymorphism:
\begin{lstlisting}[language=Dotty]
def collect[T, C^](fs: Set[box Future[T]^{C}]): Stream[box Future[T]^{C}]
\end{lstlisting}
While this works with \lstinline|C^| declaring a \emph{capture parameter}, relying on explicit polymorphism alone
would undermine CT's lightweight design.
From a language design perspective,
we would like to keep explicit polymorphism optional
and offer ergonomic alternatives for such a common pattern.

\vspace{-6pt}
\pparagraph{Reach Capabilities: Existentials without the Clutter}
We propose \emph{reach capabilities} as an effective and lightweight means to name existential captures in boxes.
With reach capabilities, the \lstinline|collect| signature becomes:
\begin{lstlisting}[aboveskip=1pt,belowskip=1pt,mathescape=true]
def collect[T](@use$\;$fs:$\;$Set[Future[box$\;$T]^]):$\;$Stream[box$\;$Future[T]^{fs*}]$\;$//$\;$<- precise capture {fs*}
\end{lstlisting}
This signature tracks the futures in \lstinline|fs| through:
(1) the reach capability \lstinline|{fs*}| that names \emph{what's in the box} of \lstinline|fs|'s type;
and (2) the \lstinline|@use| annotation signifies that the reach capability is used by \lstinline|collect| (see \Cref{sec:use-annotation}).
No extra universal quantifiers or existential types are inflicted upon users.

\vspace{-6pt}
\pparagraph{A New Foundation for Capturing Types}
Previous attempts at supporting naming mechanisms for box contents in Scala 3
\cite{DBLP:conf/programming/XuO24} suffered from several soundness issues
\cite{soundness1,soundness2,soundness3}.
This experience, along with our analysis of \lstinline|cap|'s limitations,
motivated us to develop a new theoretical foundation for CT.
We present two calculi:
\textbf{System \capless}, a new foundational capture calculus with explicit universal and existential capture quantification, and
\textbf{System \cappy}, a surface calculus formalizing CT's lightweight syntax with reach capabilities.
System \capless provides the new theoretical bedrock for capture tracking,
while System \cappy maintains CT's lightweight design.
A type-preserving translation from \cappy to \capless
assigns precise meaning to the surface syntax.
System \capless also informs our new quantifier-based capture checker implementation for Scala 3.

\vspace{-6pt}
\pparagraph{Scala Collections, Capture Checked}
The new and improved capture-checker implementation
based on System Capless finally enables integrating capturing types into Scala's entire standard collections library with minimal modifications (<5\% LoC changed, almost 90\% function signatures stay the same). 
The kinds of modifications typically look as follows:
\begin{lstlisting}[aboveskip=1pt,belowskip=1pt,keepspaces=true]
class Set[T]: // a mutable set
  def filterInPlace(pred: T => Boolean): this.type   // no changes
  def prependAll(items: IterableOnce[T]^): this.type // extra universal capture set `^`
class Iterator[T]:
  def map(f: T => Boolean): Iterator[T]^{this, f}    // capture set {this, f} on return type
\end{lstlisting}
When the lightweight notation falls short, e.g., for mutable builders (\Cref{sec:mutable-collectors}), we support \emph{optional}
explicit capture parameters. Notably, the collections library
required \emph{none}! %
Thus, our work is a decisive step towards bringing practical effect systems to real-world programming languages.

\vspace{-6pt}
\pparagraph{Contributions} To summarize, our contributions are as follows: %
\begin{itemize}[leftmargin=1.25em,nosep]

	\item \emph{Reach Capabilities:} We motivate reach capabilities (\Cref{sec:motivation}), which enable
		  expressive and lightweight effect polymorphism over generic data structures. 
		  We explain the subtleties of the previous system's boxing mechanism and universal
		  capability that necessitate reach capabilities.
	\item \emph{System \capless:} We present a new foundation for capturing types (System \capless, \Cref{sec:capless}) with existential and universal quantification of capture sets.
	Beyond providing the theoretical basis of reach capabilities,
	it is a more principled and expressive formalization of capture checking compared to the previous system \ccformal{}~\cite{DBLP:journals/toplas/BoruchGruszeckiOLLB23}.
	\item \emph{System \cappy:} We present System \cappy (\Cref{sec:cappy}) which formalizes the surface language of capture checking with rcaps,
	whose semantics is defined by a type-preserving translation to System \capless.
	\item \emph{Mechanized Metatheory:} We establish the type soundness and scope safety of our capture tracking system (\Cref{sec:metatheory}).
	      The metatheory of System \capless{} is mechanized in Lean 4, while pencil-and-paper proofs of the type-preserving translation from System \cappy to \capless are provided in the supplementary material.
	\item \emph{Implementation and Evaluation:} We applied the theory in a full re-implementation of Scala 3's capture checker %
	      (\Cref{sec:evaluation}) which was used
	      to compile capture-checked versions of an asynchronous programming library (\Cref{sec:case_study}),
		  and Scala's standard collections library.
      We assess the required changes to the latter's \textasciitilde 30K-line code base in Section \ref{sec:evaluation}.
      The required change set is simple and small enough to make the transition to capture checking practical.

\end{itemize}
Finally, we discuss limitations and future directions in \Cref{sec:limitations}, related work in Section \ref{sec:related} and conclude in Section \ref{sec:conclusion}.
The Lean 4 mechanization, the compiler implementation and the code for evaluation are available in the accompanying artifact \cite{artifact}.

\section{A Tale of Names and Boxes}\label{sec:motivation}

This section motivates \emph{reach capabilities} and the two proposed calculi.
All examples can be compiled by our implementation which is part of the Scala 3 compiler.

\subsection{A Brief Introduction to Capture Tracking}

Bringing effect tracking to a well-established, mainstream language like Scala poses specific
constraints. Scala’s broad adoption makes the ecosystem highly sensitive to notational overhead and
backward compatibility: systems that demand substantial syntactic changes or disrupt established
idioms are unlikely to be viable. Classical type-and-effect systems face a fundamental propagation
problem: effects flow along call chains, forcing each function to account for the effects of its
callees. To cope, such effect systems typically choose between manual specialization for fixed effect classes
(duplication) or pervasive effect annotations (syntactic burden). This combination of propagation
and notation has been a major barrier to deploying effect systems in Scala~\cite{DBLP:journals/toplas/BoruchGruszeckiOLLB23}.

The capability-based approach circumvents this fundamental issue
by modeling effects through capabilities tracked in the type system.
Rather than explicitly tracking effect propagation,
capabilities naturally flow through
the program as ordinary program variables.
Nevertheless, capability systems face the problem of \textit{captures} \cite{DBLP:journals/toplas/BoruchGruszeckiOLLB23},
where closures can ``leak'' effects outside of their designated lifetime
by simply holding a reference to such capabilities.
Capturing types (CT)
\cite{DBLP:journals/toplas/BoruchGruszeckiOLLB23,DBLP:conf/scala/OderskyBBLL21}
takes this capability-oriented approach while proposing a lightweight mechanism to track captures,
bringing effective-yet-ergonomic effect tracking to Scala.
A capturing type tracks the capabilities a value can capture and
takes the form of
$T\capt\set{x_1,\cdots,x_n}$, consisting of two components:
(1) the \textbf{shape type} $T$, a ``classical'' type describing the shape of the value (e.g. \lstinline|Int|, a function from \lstinline|Int| to \lstinline|Int|, etc.), and
	(2) the \textbf{capture set} $\set{x_1,\cdots,x_n}$, a set of program variables a value of this type can at most capture.
Consider the function below which prints a greeting, using the
capability \smallcode{console} for console I/O:
\begin{lstlisting}{language=Dotty}
def sayHi(name: String): Unit = console.log(s"Hi, $name!")
\end{lstlisting}
\lstinline|sayHi| has the capturing type
\lstinline|(String -> Unit)^{console}|. Since capabilities are
represented as variables, the capture set indicates the effects and resources a value of this type
can produce and access. Here, the capture set of \smallcode{sayHi} indicates that the function
\emph{at most} performs console I/O.

\subsubsection{Lightweight Effect Polymorphism}

CT supports \emph{implicit effect polymorphism}.
For higher-order functions, classical %
effect systems
have to use explicit effect binders to track effects
like in the list \lstinline|map| example from \Cref{sec:introduction}.

At its core, the issue with effect binders is about \emph{naming}.
When writing a higher-order
function, a name is needed to account for the effect produced by the argument. Failing to do so
leads to either a restrictive or an imprecise type. 

By contrast, in CT, the signature of the list \smallcode{map} method stays unchanged:
\begin{lstlisting}[language=Dotty]
trait List[+A] { def map[B](f: A => B): List[B] }
\end{lstlisting}
It is effect-polymorphic yet stays \emph{identical} to the original signature.
Under the hood, \lstinline|A => B|
expands to \lstinline[language=Dotty]|A ->{cap} B|
which itself is a shorthand for \lstinline|(A -> B)^{cap}|.
This type is a function from $A$ to $B$ capturing at most \lstinline|{cap}|
where \lstinline|cap| is the \emph{universal capability}.

\subsubsection{Universal Capability as a Device for Effect Polymorphism}
\label{sec:cap-for-polymorphism}

In CT, every capability is derived from a set of existing ones, forming a hierarchy of authority.
The universal capability \lstinline|cap| is the root of this \emph{capability hierarchy}.
E.g., the following
\lstinline|Logger| class
writes logging messages to both a file and the console:
\begin{lstlisting}[language=Dotty]
class Logger(f: File^) { def log(msg: String): Unit = { f.write(msg); console.log(msg) } }
val f: File^ = ...
val logger: Logger^{f,console} = Logger(f)
\end{lstlisting}
Here, \lstinline|File^| (short for \lstinline|File^{cap}|) is a capability for file I/O.
The \lstinline|logger| capability
obtains access to the file and the console
from existing capabilities \lstinline|f| and \lstinline|console|,
i.e., deriving from \lstinline|f| and \lstinline|console|.
Furthermore, the \lstinline{sayHi} function
whose type is \lstinline|String ->{console} Unit|
can be viewed as a capability derived from \lstinline|console|.

CT introduces \emph{subcapturing},
a subtyping relation between capture sets.
It augments set inclusion with the capability hierarchy.
A capability is a subcapture of the capabilities it derives from.
For instance, the following subcapturing relations hold for the \lstinline|Logger| example:
\begin{lstlisting}[columns=fixed]
         {} <: {logger} <: {f,console}         {} <: {f} <: {cap}
\end{lstlisting}
Subtyping between capturing types is defined by
the combination of
regular subtyping between shape types and subcapturing.
Since all capabilities are ultimately derived from the universal capability,
any capture set is a subcapture of \lstinline|{cap}|.
Therefore, subcapturing and the universal capability \lstinline|cap| can be used as a device for \emph{effect polymorphism}.
Going back to \lstinline|List.map|,
the argument type \lstinline|A => B| indicates that \lstinline|map| takes functions performing arbitrary effects:
it is effect-polymorphic.
This is analogous to using \lstinline|Any| as the top type for subtype polymorphism.

In fact,
the argument \lstinline|f| of \lstinline|map| itself is a capability,
thus \lstinline|f| becomes a \emph{name} of its effects.
We do not need an extra name!
For example,
the following function takes an operation
and returns an iterator that repeatedly applies that operation:
\begin{lstlisting}[language=Dotty]
def repeated[T](f: () => T): Iterator[T]^{f} = new Iterator[T]:
  def next(): T = f()
  def hasNext(): Boolean = true
\end{lstlisting}
The signature reads that, given any operation \lstinline|f|,
the function returns an iterator
with \lstinline|f|'s effects.

\subsubsection{What's in the Box?}\label{sec:motiv:whats-in-the-box}
Another challenge in supporting capture tracking in Scala
is the interaction between capturing types and generics,
which is a fundamental part of functional programming.
Consider the following generic function that transforms the first element of a pair:
\begin{lstlisting}[language=Dotty]
def mapFirst[A, B, C](p: Pair[A, B], f: A => C): Pair[C, B] = Pair(f(p.x), p.y)
\end{lstlisting}
What should be the signature of this function under CT?
Without any restrictions, the type variables \lstinline|A|, \lstinline|B|, and \lstinline|C| could be instantiated to capturing types.
This means that the pair \lstinline|p| could capture arbitrary capabilities through its fields.
Consequently, the parameter \lstinline|p| needs to be annotated with \lstinline|{cap}|:
\begin{lstlisting}[language=Dotty]
def mapFirst[A, B, C](p: Pair[A, B]^{cap}, f: A => C): Pair[C, B]^{cap} = ...
\end{lstlisting}
The resulting \lstinline|Pair| also needs to be annotated with \lstinline|{cap}| because there is no account of what \lstinline|C| captures: it can be anything.
This is an unacceptably imprecise signature.

The solution in CT is to keep generic types pure
and introduce boxes to recover expressiveness.
This simplifies generic functions, as they do not need to account for potential effects from their type arguments.
With this restriction, the \lstinline|mapFirst| function works without any capture annotations:
\begin{lstlisting}[language=Dotty]
def mapFirst[A, B, C](p: Pair[A, B], f: A => C): Pair[C, B] = ...
\end{lstlisting}
Intuitively,
parametricity ensures that \lstinline|mapFirst| cannot inspect the contents of the pair,
so the captures of the generic field types should be irrelevant to this function.

To handle capturing types in generic contexts,
CT introduces \emph{boxes}, which encapsulate impure values as pure ones.
To access a boxed value,
it needs to be \emph{unboxed},
which ``pops out'' the captures that were previously hidden.
For example:
\begin{lstlisting}[language=Dotty]
val consoleOps: List[box () ->{console} Int]^{} = List(box () => console.readInt())
val f: () ->{} Boolean = () => consoleOps.isEmpty
val g: () ->{console} Int = () => (unbox consoleOps.head)()
\end{lstlisting}
Even though elements of \lstinline|consoleOps| are effectful,
they are all boxed and the list \lstinline|consoleOps| is pure.
The function \lstinline|f| is pure
since it does not access the elements of \lstinline|consoleOps|.
Conversely, \lstinline|g| \emph{does} access the elements of \lstinline|consoleOps|
and unboxes its element,
which pops out the captured capabilities hidden by the box.
Therefore,
\lstinline|g| captures the previously hidden capability \lstinline|console|.

As a rule of thumb,
whenever we see a capturing type in type-argument position,
there is implicitly a box.
For instance, the type \lstinline|List[() => Int]|
expands automatically to \lstinline|List[box () => Int]|.
We will nevertheless show boxes in examples for pedagogical reasons.
In fact, the Scala 3 compiler implements complete box inference,
so boxes are transparent to users \cite{Xu2023FormalizingBI} and
the language does not even have a surface syntax for them.\footnote{Unlike the ``boxing'' in the JVM,
boxes in CT are purely a compile-time construct
with no runtime overhead.}

Now we can see how the \lstinline|mapFirst| function works with boxed impure values.
Consider a pair containing a file capability and an operation:
\begin{lstlisting}[language=Dotty,keepspaces=true]
val p: Pair[box File^{f}, box () ->{f} Unit] = ...
val q = mapFirst(p, f => box (new Logger(unbox f))) // : Pair[box Logger^{f}, box () ->{f} Unit]
val useLogger = () => (unbox q.fst).log("test")     // : () ->{f} Unit
\end{lstlisting}
The generic function \lstinline|mapFirst| operates on this pair with all the impure values boxed.
The transformation function creates a new \lstinline|Logger| by unboxing the file,
and the result is re-boxed to maintain purity.
When the logger is finally used,
unboxing it reveals the capability \lstinline|f| in the capture set.
This demonstrates the idea of \emph{capture tunnelling} \cite{DBLP:journals/toplas/BoruchGruszeckiOLLB23}:
captures are tunneled through generic contexts via boxes
and only surface when the boxed values are accessed.
This reflects the relational parametricity \cite{parametricity} of generic functions,
and allows CT to stay concise and practical \cite{DBLP:journals/toplas/BoruchGruszeckiOLLB23}.

Furthermore, boxes play a crucial role in ensuring the scope safety of capabilities:
a boxed value capturing \lstinline|cap| (e.g. \lstinline|box File^{cap}| where \lstinline|^| binds tighter than the box) cannot be unboxed,
as it typically represents a capability that has escaped from its defining scope \cite{DBLP:journals/toplas/BoruchGruszeckiOLLB23}.
The following example tries to leak a \lstinline|File| out of its defining local scope:
\begin{lstlisting}[language=Dotty,keepspaces=true]
def withFile[T](f: File^ => T): T = { val l: File^ = new File; val r = f(l); l.close(); r }
val leakedBox: box File^{cap} = withFile[box File^{cap}](file => box file)
val leakedFile = unbox leakedBox // error: unboxing value capturing {cap}
\end{lstlisting}
Note the type parameter to \lstinline|withFile| has to be instantiated with \lstinline|box File^{cap}|, as
the local parameter \lstinline|file| is out of scope. Consequently, the unbox operation fails due to the above restriction.

Despite these properties, boxes introduce a fundamental limitation: they \emph{cut the tie} between
the name of a generic data structure
and the effects of its elements, making them untrackable.

For instance,
given \lstinline|ops: List[() => Int]|,
we cannot use \lstinline|ops| to name the effects of the list elements:
\begin{lstlisting}[language=Dotty]
def mkIterator[T](ops: List[() => T]): Iterator[T]^{cap} = ...
\end{lstlisting}
Here, \lstinline|mkIterator| creates an iterator from a list of closures, running them one by one.
This function cannot be expressed in the previous CT system \cite{DBLP:journals/toplas/BoruchGruszeckiOLLB23}
due to the scope safety restriction mentioned above.
Furthermore,
even if we had sacrificed scope safety and lifted the restriction,
we are only able to type the result at \lstinline|Iterator[T]^{cap}|
since we have no means to \emph{name} the effects of the elements in \lstinline|ops|.
This is again utterly imprecise:
even if only pure operations are passed in,
the result is considered performing arbitrary effects. The following definition of
\lstinline|pure| has a pure RHS, but will fail to type-check:
\begin{lstlisting}[language=Dotty,mathescape=true]
val pure: () ->{} Int = () => mkIterator(List(() => 1)).next()$\;$//$\;$error:$\;$using value capturing {cap}
\end{lstlisting}
The root cause lies in the \lstinline|mkIterator| example itself.
The argument \lstinline|ops| (of type \lstinline|List[box () => T]^{}|)
is pure: \lstinline|{ops} <: {}|,
since list elements are \textit{boxed}.
Therefore, the list's capture set becomes completely disconnected from the capture sets of its elements.
Hence, we cannot name \emph{what's in the box} of a generic data structure!
This inability to characterize the contents of boxes is precisely the underlying problem that this paper addresses.

\subsection{Naming What's in the Box}\label{sec:motiv:reach-capabilities}

The problem of naming capabilities inside boxes can be solved by extending the type system with explicit quantification over capture sets.
This is supported by our new foundational calculus,
System \capless{},
which provides a sound and principled basis for explicit capture quantifications.
We can give a precise type to our \lstinline|mkIterator| example:
\begin{lstlisting}[language=Dotty]
def mkIteratorExplicit[T, C^](ops: List[() ->{C} T]): Iterator[T]^{C} = ...
\end{lstlisting}
Here, \lstinline|[..., C^]| introduces a universal capture set variable \lstinline|C|.
The signature now precisely states that for any capture set \lstinline|C|,
if \lstinline|mkIterator| is given a list of operations that all at most capture \lstinline|C|,
it returns an iterator that also captures \lstinline|C|.
\begin{lstlisting}[language=Dotty,keepspaces=true]
mkIteratorExplicit[Int, {}](List(() => 1, () => 2)) // : Iterator[Int]^{}
mkIteratorExplicit[Int, {console}](consoleOps)      // : Iterator[Int]^{console}
\end{lstlisting}
When called with pure operations, \lstinline|c| can be instantiated to the empty set \lstinline|{}|, and the resulting iterator is pure.
When called with \lstinline|consoleOps| of type \lstinline|List[() ->{console} Int]|, \lstinline|c| is instantiated to \lstinline|{console}|, and the result captures \lstinline|{console}|.

While expressive,
this explicit style can be verbose:
every function that maps generic collections of capabilities
has to be annotated with explicit capture variables.
This deviates from the lightweight philosophy
that makes CT appealing for established languages like Scala.
We therefore propose to
keep explicit quantification as an optional feature,
and introduce \emph{reach capabilities},
an ergonomic and lightweight mechanism for naming ``what's in the box''.
Reach capabilities allow us
to write the \lstinline|mkIterator| signature
as follows,
with minimal disruption to the original code:
\begin{lstlisting}[language=Dotty]
def mkIterator[T](@use ops: List[() => T]): Iterator[T]^{ops*} = ...
\end{lstlisting}
Reach capabilities can be understood by translating them to explicit capture variables.
For instance, the reach capability \lstinline|ops*| directly corresponds to
the variable \lstinline|c| in the explicit version.
It serves as a \emph{name} for the capabilities that can be \emph{reach}ed through the boxes of \lstinline|ops|.
Similar to the explicit version, the following is well-typed:
\begin{lstlisting}[language=Dotty,keepspaces=true]
mkIterator(List(() => 1, () => 2)) // : Iterator[Int]^{}
mkIterator(consoleOps)             // : Iterator[Int]^{console}
\end{lstlisting}
Reach capabilities are realized through three core mechanisms: reach refinement (\Cref{sec:rcaps:reach-refinement}), deep capture sets (\Cref{sec:rcaps:deep-capture-sets}), and the \lstinline|@use| annotation (\Cref{sec:use-annotation}).

\subsubsection{Reach Refinement}\label{sec:rcaps:reach-refinement}
To introduce rcaps,
the type-checker performs \emph{reach refinement}. When a variable \lstinline|ops| is used, this
process replaces certain occurrences of the universal capability \TCAP{} in its type with the reach
capability \lstinline|ops*|. This essentially gives a name to the capabilities inside the boxes,
which is analogous to how a variable names the capabilities it directly captures. For instance,
given \lstinline|ops: List[() ->{cap} T]|, reach refinement infers its type as \lstinline|List[() ->{ops*} T]|. 
Let's inspect the \lstinline|mkIterator| function as an example:
\begin{lstlisting}[language=Dotty]
def mkIterator[T](@use ops: List[box () => T]): Iterator[T]^{ops*} = new Iterator[T]:
  var current: List[box () ->{ops*} T] = ops
  def next(): T =
    val f: () ->{ops*} T = unbox (current.head : box () ->{ops*} Unit)
    current = current.tail
    f()
  def hasNext(): Boolean = current.nonEmpty
\end{lstlisting}
This type-checks thanks to reach refinement:
\begin{itemize}[leftmargin=*]
	\item In the definition of \lstinline|current|, \lstinline|ops| is accepted because its type is refined from \lstinline|List[box () ->{cap} Int]| to \lstinline|List[box () ->{ops*} Int]|. %
	\item On the RHS of the variable definition \lstinline|f|, \lstinline|current.head| has type \lstinline|box () ->{ops*} Int|.
	\item The unboxing then propagates the reach capability \lstinline|ops*| to the capture set of the iterator's closure.
	\item The resulting iterator correctly captures \lstinline|ops*|.
\end{itemize}
From an explicit-quantification perspective,
the parameter \lstinline|ops| has the type \lstinline[mathescape=true]|$\exists c.$ List[box () ->{$c$} T]|,
and \lstinline|ops*| corresponds directly to the witness of $c$.

Understanding rcaps
in terms of a translation to quantification
is essential for a sound design.
The earlier,
ad-hoc implementation of rcaps suffered from soundness issues
precisely because it lacked this foundation \cite{soundness1,soundness2,soundness3}.
The central question is:
which occurrences of the universal capability \TCAP{} in a type should be refined to a reach capability?
The initial, intuitive answer was ``all covariantly-occurring \TCAP{}s''.
This is unsound, accepting the following code:
\begin{lstlisting}[language=Dotty]
val map: [T] -> (files: List[box File^]) -> (op: (box File^) => T) -> List[T] =
  files.map(op) // a function that maps a list of files
val makeFilePure: File^ -> File^{map*} = (f: File^) => map[box File^{map*}](List(f))(x => x).head
\end{lstlisting}
The \lstinline|makeFilePure| function is problematic:
it converts
an arbitrary \lstinline|File^| capability (which can perform file I/O)
to one that only captures \lstinline|map*|.
Since \lstinline|map| is a pure function, \lstinline|map*| is empty,
meaning the resulting \lstinline|File| is considered \emph{pure},
despite being the same impure \lstinline|File| that was given!

The problem lies within
\lstinline|map|'s reach refinement:
\begin{lstlisting}
  [T] -> (files: List[box File^{cap}]) -> (op: (f: box File^{map*}) => T) -> T
\end{lstlisting}
It becomes clear that this is absurd when we translate the type of \lstinline|map| into explicit quantification:\footnote{For clarity, we show existential quantifiers on parameters, which are trivially convertible into universal quantifiers.}
\begin{lstlisting}[mathescape=true]
  [T] -> (files: $\exists c_1.$ List[box File^{$c_1$}]) -> (op: (f: $\exists c_2.$ box File^{$c_2$}) => T) -> T
\end{lstlisting}
\lstinline|map*| witnesses the existential scoped at the outermost level of \lstinline|map| (in this case, there is no existential bound at this level),
and clearly does not witness $c_2$.
The unsound reach refinement principle implicitly assumes \lstinline|map|'s type translates to:
\begin{lstlisting}[mathescape=true]
  $\exists c_2.$ [T] -> (files: $\exists c_1.$ List[box File^{$c_1$}]) -> (op: (f: box File^{$c_2$}) => T) -> T
\end{lstlisting}
It confuses the scope of the existential variable $c_2$.
Our principled approach avoids such confusion
by ensuring that each \lstinline|cap| is mapped to an existential in the closest enclosing scope,
and reach capabilities always witness the outermost existential.

\subsubsection{Deep Capture Sets}\label{sec:rcaps:deep-capture-sets}
At function call sites, reach capabilities of parameters are instantiated with the \emph{deep capture sets} of the corresponding argument types.
The deep capture set of a type collects all capabilities that occur covariantly in that type.
It provides a concrete witness for the existential quantification that reach capabilities represent.

Let us revisit our \lstinline|mkIterator| example with its explicit quantification version:
\begin{lstlisting}[language=Dotty]
def mkIteratorExplicit[T, c^](ops: List[() ->{c} T]): Iterator[T]^{c} = ...
\end{lstlisting}
When calling \lstinline|mkIteratorExplicit[Int, {console}](consoleOps)|, 
the capture variable \lstinline|c| is explicitly instantiated to \lstinline|{console}|.
Deep capture sets achieve the same for reach capabilities:
when calling \lstinline|mkIterator(consoleOps)|,
the reach capability \lstinline|ops*| (which witnesses the explicit variable \lstinline|c|) is instantiated with the deep capture set of \lstinline|List[() ->{console} Int]|,
which yields \lstinline|{console}|, correctly typing the result as \lstinline|Iterator[Int]^{console}|.
Deep capture sets collect only covariantly-occurring capture sets because these represent capabilities
that can flow ``outward'' -- the capabilities ``in the box'' that can be accessed when using a value.
Contravariant positions, by contrast, represent capabilities flowing ``inward'' -- requirements for
using the value rather than capabilities it captures.

\subsubsection{The \lstinline|@use| Annotation}\label{sec:use-annotation}
The \lstinline|@use| annotation on a function parameter
signifies that the parameter's reach capability is \emph{used} by the function.
This annotation is needed to ensure that function applications are tracked correctly.
Normally, the capabilities captured by an application \lstinline|f(x)| are simply \lstinline|{f,x}|.
However, this is unsound for functions that use their parameters' reach capabilities:
\begin{lstlisting}
def runOps(@use ops: List[() => Unit]): Unit = ops.foreach(op => op()) // run each op in list
val ops: List[() ->{console} Unit] = List(() => console.log("Hello"))
val r2 = () => runOps(ops)
\end{lstlisting}
Without a special rule for \lstinline|@use| parameters,
the body of \lstinline|r2| would only capture \lstinline|{runOps, ops}|,
which is pure.
As a result,
\lstinline|r2| would be incorrectly typed as \lstinline|() ->{} Unit|, even though it performs console I/O.
The \lstinline|@use| annotation signals that the call site must account for the capabilities \emph{inside} the argument.
The general rule is that for \lstinline|@use| parameters, the call captures the deep capture set of the argument, whereas for normal parameters only its (shallow) capture set is used.
For \lstinline|runOps(ops)|, the captured set is the union of \lstinline|{runOps}| and the deep capture set of \lstinline|ops|'s type, which is \lstinline|{console}|.
This correctly gives \lstinline|r2| the type \lstinline|() ->{console} Unit|.

\subsubsection{Type Definitions}\label{sec:motiv:typedefs}
Sometimes, the default behavior of introducing existential quantifiers in the closest
enclosing scope is not desirable.
To address this,
System \cappy{} includes \emph{type definitions},
which are essentially parameterized aliases for types.
For instance,
the following type definition
defines non-dependent functions, i.e., those that do not depend on their parameters:
\begin{lstlisting}[language=Dotty]
type Function[-A, +B] = (z: A) -> B
\end{lstlisting}
Instances of \lstinline|cap| in a type definition's type parameters are translated into existentials \emph{before} the type definition is expanded.
Consider the \lstinline|flatMap| function for \lstinline|Iterator| in the Scala collections library:
\begin{lstlisting}[language=Dotty,mathescape=true]
def flatMap[A, B](it: Iterator[A]^, f: A => Iterator[B]^): Iterator[B]^{?} //$\;$<- what's the result?
\end{lstlisting}
With the default translation scheme,
the resulting iterator's capture set is the overly imprecise \lstinline|{cap}|.
This is because type parameter \lstinline|f|'s type translates to
\lstinline[mathescape=true]|A => $\exists c.$ Iterator[B]^{c}|.
Each application of \lstinline|f| yields an iterator that captures a locally-quantified existential $c$.
A more desirable translation would be
\lstinline[mathescape=true]|$\exists c.$ A => Iterator[C]^{c}|,
with the existential $c$ being scoped over the whole function,
and the reach capability \lstinline|f*| would be the capture set of the result iterator.
By treating non-dependent function types \lstinline|T => U|
as applied type definitions \lstinline|Function[T, U]^|,
we can type-check the following signature with a precise result capture set:
\begin{lstlisting}[language=Dotty]
def flatMap[A, B](it: Iterator[A]^, f: A => Iterator[B]^): Iterator[B]^{it, f, f*}
\end{lstlisting}
This is exactly how the implementation works.
The applied type \lstinline|Function[A, Iterator[B]^]|
translates and dealiases as follows:
\begin{lstlisting}[mathescape=true]
Function[A, Iterator[B]^{cap}] $\leadsto$ $\exists c.$ Function[A, Iterator[B]^{c}] $\leadsto$ $\exists c.$ (z: A) -> Iterator[B]^{c}
\end{lstlisting}
Type definitions play a crucial role in our system
by offering a way to change the scope of existentials,
as needed by, e.g., church-encoded data structures.
We further discuss them in Section~\ref{sec:cappy:type-definitions}.

While explicit quantification over capabilities provides a more powerful and general solution to capability polymorphism,
this generality often comes at the cost of verbosity and boilerplate.
Reach capabilities,
in contrast,
offer a lightweight and ergonomic alternative that is sufficient for the vast majority of common programming patterns.
For instance,
rcaps are sufficiently expressive for capture-checking the standard library.
We do not need any explicit quantification.
The usefulness of reach capabilities is further demonstrated in the case study on asynchronous programming (Section~\ref{sec:case_study}) and in\condtext{}{ our technical report} \appendixref{B.1}{sec:local-mutable-state}.
\section{System \capless{}: A New Foundation for Expressive Capture Tracking}
\label{sec:capless}

\begin{wide-rules}\noindent
	{\footnotesize\begin{multicols}{3}\noindent
		\begin{flalign*}
			x,\,y,\,z         \tag*{\textbf{Variable}}\\
			X                 \tag*{\textbf{Type Variable}}\\
			\new{c}                 \tag*{\textbf{Capture Variable}}\\
			s,\,t,\,u\coloneqq\ &           \tag*{\textbf{Term}}\\
			&a                              \tag*{answer}\\
			&x\,y                              \tag*{app.}\\
			&x[S]                              \tag*{type app.}\\
			&\new{x[c]}                              \tag*{capture app.}\\
			&\LET x = t \IN u                  \tag*{let}\\
			&\new{\LET \langle c, x \rangle = t \IN u}    \tag*{existential let}\\
			v\coloneqq\ &           \tag*{\textbf{Value}}\\
			&\lambda(x: T)t                      \tag*{term lambda}\\
		\end{flalign*}
		\begin{flalign*}
			&\lambda[X<:S]t                      \tag*{type lambda}\\
			&\new{\lambda[c<:B]t}                    \tag*{capt. lambda}\label{syn:capt-lambda}\\
			&\new{\langle C, x\rangle}                    \tag*{pack}\\
			E,\,F\coloneqq\ &           \tag*{\textbf{Existential Type}}\\
			&\new{\exists c.\, T}   \tag*{existential}\\
			&T   \tag*{type}\\
			R,\,S\coloneqq\ &           \tag*{\textbf{Shape Type}}\\
			&\top                    \tag*{top}\\
			&X                    \tag*{type variable}\\
			&\forall(x: T)E             \tag*{term function}\\
			&\forall[X<:S]E             \tag*{type function}\\
			&\new{\forall[c<:B]E}             \tag*{capt. function}\label{syn:capt-fun}\\
   \end{flalign*}
   \begin{flalign*}
			a\coloneqq\ & x \mid v           \tag*{\textbf{Answer}}\\
			\theta\coloneqq\ & x \mid \new{c}           \tag*{\textbf{Capture}}\\
			C,\,D\coloneqq\ &\set{\theta_1,\cdots,\theta_n}     \tag*{\textbf{Capture Set}}\\
			B\coloneqq\ & * \mid C           \tag*{\textbf{Capture Bound}}\\
			T,\,U\coloneqq\ &           \tag*{\textbf{Type}}\\
			&S\capt C                    \tag*{capturing}\\
			&S                    \tag*{pure}\\
			\G,\,\Delta\coloneqq\ &           \tag*{\textbf{Context}}\label{syn:context}\\
			&\emptyset                    \tag*{empty}\\
			&\G, x: T                    \tag*{term binding}\label{syn:term-binding}\\
			&\G, X<:S                    \tag*{type binding}\\
			&\new{\G, c<:B}                    \tag*{capt. binding}\\
	\end{flalign*}
	\end{multicols}}
	\vspace{-3em}
	\caption{Abstract syntax of System \capless{}. Key differences from System \ccformal{} are \tnew{highlighted}.}\label{fig:syntax}
\end{wide-rules} 
\begin{figure*}[htbp]
\footnotesize

\flushleft{\textbf{Typing \quad $\typ{C}{\G}{t}{E}$}}

\vspace{-1em}

\begin{multicols}{3}

\infrule[\ruledef{var}]
{x: S\capt C\in \G}
{\typ{\set{x}}{\G}{x}{S\capt{\set{x}}}}

\infrule[\ruledef{pack}]
{\typ{C'}{\G}{x}{[c:=C]T}}
{\typ{\set{}}{\G}{\langle C, x\rangle}{\EXCAP{c} T}}

\infrule[\ruledef{sub}]
{\typ{C}{\G}{t}{E}\andalso
 \subs{\G}{E}{F}\\
 \subs{\G}{C}{C'}\andalso
 \wf{\G}{C',F}}
{\typ{C'}{\G}{t}{F}}

\infrule[\ruledef{abs}]
{\typ{C}{(\G, x: T)}{t}{E}\andalso
 \wf{\G}{T}}
{\typ{\set{}}{\G}{\lambda(x: T)t}{(\forall(x: T) E)\capt \left(C\setminus x\right)}}

\infrule[\ruledef{app}]
{\typ{C'}{\G}{x}{(\forall(z: T) E)\capt C}\\
 \typ{C'}{\G}{y}{T}}
{\typ{C'}{\G}{x\,y}{[z:=y]E}}

\infrule[\ruledef{tabs}]
{\typ{C}{(\G, X<:S)}{t}{E}\andalso\wf{\G}{S}}
{\typ{\set{}}{\G}{\lambda[X<:S]t}{(\forall[X<:S] E)\capt C}}

\infrule[\ruledef{tapp}]
{\typ{C'}{\G}{x}{(\forall[X<:S] E)\capt C}}
{\typ{C'}{\G}{x[S]}{[X:=S]E}}

\infrule[\ruledef{cabs}]
{\typ{C}{(\G, c<:B)}{t}{E}\andalso \wf{\G}{C}}
{\typ{\set{}}{\G}{\lambda[c<:B]t}{(\forall[c<:B]E)\capt C}}

\infrule[\ruledef{capp}]
{\typ{C'}{\G}{x}{(\forall[c<:D]E)\capt C}}
{\typ{C'}{\G}{x[D]}{[c:=D]E}}

\infrule[\ruledef{let}]
{\typ{C}{\G}{t}{T}\andalso
 \typ{C}{(\G, x: T)}{u}{E}\\ 
 \wf{\G}{C, E}}
{\typ{C}{\G}{\LET x = t\IN u}{E}}

\infrule[\ruledef{let-e}]
{\typ{C}{\G}{t}{\EXCAP{c} T}\\
 \typ{C}{(\G, c<:*, x: T)}{u}{F}\\ 
 \wf{\G}{C, F}}
{\typ{C}{\G}{\LET \langle c, x\rangle = t\IN u}{F}}

\end{multicols}

\vspace{-1.5em}
\flushleft{\textbf{Subcapturing \quad $\subs{\G}{C_1}{C_2}$}}

\begin{multicols}{5}

\infrule[\rruledef{sc-trans}]
{\subs{\G}{C_1}{C_2}\\
 \subs{\G}{C_2}{C_3}}
{\subs{\G}{C_1}{C_3}}

\infrule[\rruledef{sc-var}]
{x : S\capt C \in \G}
{\subs{\G}{\set{x}}{C}}

\infrule[\rruledef{sc-bound}]
{c <: C\in \G}
{\subs{\G}{\set{c}}{C}}

\infrule[\rruledef{sc-elem}]
{C_1\subseteq C_2}
{\subs{\G}{C_1}{C_2}}

\infrule[\rruledef{sc-set}]
{\subs{\G}{C_1}{C}\\ 
 \subs{\G}{C_2}{C}}
{\subs{\G}{C_1\cup C_2}{C}}

\end{multicols}

\vspace{-1.6em}
\flushleft{\textbf{Bound Subtyping \quad $\subs{\G}{B_1}{B_2}$} \text{same as subcapturing plus} $\subs{\G}{B}{*}$}

\flushleft{\textbf{Subtyping \quad $\subs{\G}{E_1}{E_2}$}}

\vspace{-1em}

\begin{multicols}{4}

\infax[\ruledef{top}]
{\subs{\G}{S}{\top}}

\infax[\ruledef{refl}]
{\subs{\G}{E}{E}}

\infrule[\ruledef{trans}]
{\subs{\G}{E_1}{E_2}\\ \subs{\G}{E_2}{E_3}}
{\subs{\G}{E_1}{E_3}}

\infrule[\ruledef{tvar}]
{\\
 X<:S\in\G}
{\subs{\G}{X}{S}}

\infrule[\ruledef{capt}]
{\subs{\G}{S_1}{S_2}\\\subs{\G}{C_1}{C_2}}
{\subs{\G}{S_1\capt C_1}{S_2\capt C_2}}
  
\end{multicols}

\begin{multicols}{2}

\infrule[\ruledef{exist}]
{\subs{(\G, c<:*)}{T_1}{T_2}}
{\subs{\G}{\exists c.\, T_1}{\exists c.\, T_2}}

\infrule[\ruledef{fun}]
{\subs{(\G, x: T_2)}{E_1}{E_2}
 \andalso\subs{\G}{T_2}{T_1}}
{\subs{\G}{\forall(x: T_1) E_1}{\forall(x: T_2) E_2}}

\infrule[\ruledef{tfun}]
{\subs{(\G, X<:S_2)}{E_1}{E_2}\andalso\subs{\G}{S_2}{S_1}}
{\subs{\G}{\forall[X<:S_1] E_1}{\forall[X<:S_2] E_2}}

\infrule[\ruledef{cfun}]
{\subs{(\G, c<:B_2)}{E_1}{E_2}\\
 \subs{\G}{B_2}{B_1}}
{\subs{\G}{\forall[c<:B_1] E_1}{\forall[c<:B_2] E_2}}

\end{multicols}

\vspace{-1.5em}
\caption{Typing rules of System \capless{}.}\label{fig:all-typing}

\end{figure*}
 
System \capless{} (\Cref{fig:syntax,fig:all-typing}) is a new foundation of capturing types. It
models the essence of our new capture checker implementation in Scala 3 (\Cref{sec:evaluation}).
Formally, it is a version of System \ccformal{} \cite{DBLP:journals/toplas/BoruchGruszeckiOLLB23},
the main difference being (1) the removal of the universal capability $\CAP$, and (2) having
explicit capture quantifications. 
\subsection{Syntax}\label{sec:capless-syntax}

The syntax of System \capless{} is shown in \Cref{fig:syntax}, highlighting
key differences to \ccformal{} \cite{DBLP:journals/toplas/BoruchGruszeckiOLLB23}.
Just like \ccformal{}, we represent programs in monadic normalform (MNF)~\cite{DBLP:conf/popl/HatcliffD94}.
This has the advantage that substitutions in dependent applications are always variable
renamings and preserve the structure of types.
The main difference stems from dropping the top capture set $\CAP$ in favor of explicit bounded universal
and existential capture quantification, which behave similarly to the respective quantifiers in
System F$_\leq$ for types. Accordingly, capture sets $C$ can now also mention capture variables $c$ next to
term variables $x$.

Capture quantification is bounded by a capture bound $B$ which can be either a concrete
capture set upper bound or \emph{unbounded} (denoted by $*$), though note that there is no ``top''
capture set any longer. When the bound $B$ is omitted, it defaults to $*$.

Existential capture quantification is unbounded and second class,
i.e., confined to the top level in the type system and function result types.
The typing judgment
assigns types from the syntax category $E$ to terms.
The argument types of term functions only range over the syntax category $T$.
This restriction is a deliberate design choice aimed at the minimality of the system.
It does not compromise expressiveness,
as functions that would otherwise take existential types as arguments
can always be equivalently expressed using universal capture quantification
followed by ordinary term functions.
Concretely, the type $\forall(z\!: \EXCAP{c}T)E$
is equivalently represented as
$\forall[c]\forall(z : T)E$.

\subsection{Type System}\label{sec:capless-typing}

The typing judgment $\typ{C}{\G}{t}{E}$  in \Cref{fig:all-typing} states that
term $t$ has existential type $E$ under context $\G$ with use set $C$.
Intuitively,
the use set $C$ are
the set of capabilities that will \emph{at most} be \emph{used}
by the evaluation of term $t$.
The use set of all values is empty, as they are already evaluated.
Tracking use sets
in the judgment improves over the previous System \ccformal{}
\cite{DBLP:journals/toplas/BoruchGruszeckiOLLB23} in two important ways: (1) capture tracking for
curried functions behaves more precisely akin to effect systems, i.e., captures of subsequent
function arrows do not accumulate on the current one, and (2) a more streamlined handling of let bindings.
We explain them in detail when discussing corresponding typing rules.

Typing rules (\Cref{fig:all-typing}) are mostly identical to System \ccformal{} with the
main difference
that typing potentially assigns an existential type $E$
and that there are changes and additions due to the reformulation with use sets 
and the introduction of universal and existential capture quantifications.

To type a variable $x$ in context \ruleref{var}, it has to be included in the use set.
Just like in \ccformal{}, the capture set is refined to $\{x\}$ in the assigned type
and the capture set $C$ in the environment can be recovered via subcapturing.
The subtyping rule \ruleref{sub} allows refining both the use set and the type.

Type abstraction \ruleref{tabs} and type application \ruleref{tapp} is mostly standard,
and follows System F$_\leq$, again remarking that values are pure, having an empty use set.
Following \ccformal{},
type application takes pure types without capture qualifiers.
The rules for universal \ruleref{cabs}, \ruleref{capp} and existential capture
quantification \ruleref{pack}, \ruleref{let-e} are in the spirit of System F$_\leq$ and hence
unsurprising.

Introducing and eliminating a $\lambda$-abstraction in \ruleref{abs} and \ruleref{app} is for the most part standard. In terms of captures, a $\lambda$ is pure, i.e., having an empty use set, and the use set
of the body $t$ is annotated to the function type. 
Thanks to the tracking of use sets in the typing judgment,
the capture tracking of curried functions is more precise compared to \cite{DBLP:journals/toplas/BoruchGruszeckiOLLB23}.
For instance, consider the term 
\(\small
\lambda(x_1:\textsf{Unit}).\ \LET z_0 = \textsf{logger}.\textsf{log}(\cdots)\IN\ \lambda(x_2:\textsf{Unit}).\,\textsf{console}.\textsf{readInt}()
\).
This term has type $(\forall(x_1:\textsf{Unit})(\forall(x_2:\textsf{Unit})\textsf{Int})\capt\set{\textsf{console}})\capt\set{\textsf{logger}}$,
while in the previous system \cite{DBLP:journals/toplas/BoruchGruszeckiOLLB23}
it would have been $(\forall(x_1:\textsf{Unit})(\forall(x_2:\textsf{Unit})\textsf{Int})\capt\set{\textsf{console}})\capt\set{\textsf{console},\textsf{logger}}$.
Our system has a better account of
\emph{when} the captured capabilities are used.

\subsubsection{Subtyping and Subcapturing}

Subcapturing and subtyping rules (\Cref{fig:all-typing}) follow System \ccformal{}.
The former is a preorder on capture sets that subsumes set inclusion, plus the more interesting rule \rruleref{sc-var}
which ``\emph{reflects an essential property of object capabilities}''~\cite{DBLP:journals/toplas/BoruchGruszeckiOLLB23}, namely that
the singleton capture set $\set{x}$ refines/derives from the capabilities $C$ from which $x$ was created.
For the most part, subtyping integrates subcapturing with the standard subtyping rules 
for kernel System F$_\leq$ unsurprisingly.
In addition,
the rules \ruleref{exist} and \ruleref{cfun} are for the new quantification forms.

\subsubsection{Let Bindings and Avoidance}
Let bindings \ruleref{let} are typed in a standard manner. Due to dependent typing, we need to
avoid mentions of the bound variable in the result type, which we enforce by subtyping.
Well-formedness of capture sets $\wf{\G}{C}$ and its lifting to types $\wf{\G}{E}$ 
(defined in \appendixref{A.1}{sec:well-formedness}) enforce that all mentioned variables are in scope.
Note that the locally bound variable $x$ is not dropped from the use set $C$.
Instead, it needs to be avoided in the use set with subcapturing.
This, along with preciser curried function types,
enables reasoning about \emph{when} the captured capabilities are used.
For instance, assuming $x_f : (\forall(x:\textsf{Unit})(\forall(y:\textsf{Unit})\textsf{Int})\capt\set{\textsf{console}})\capt\set{\textsf{logger}}$ and $\textsf{unit}$ a constant of type $\textsf{Unit}$,
the term $\LET z = x_f\,\textsf{unit}\IN z\,\textsf{unit}$
has the use set $\set{\textsf{console},\textsf{logger}}$.
This is because $z$, whose type is $(\forall(y:\textsf{Unit})\textsf{Int})\capt\set{\textsf{console}}$,
is mentioned in the body of the let binding.
To avoid the locally-bound $z$ in the use set,
we must widen $z$ to $\set{\textsf{console}}$.
Conversely, the following term,
which applies $x_f$ once and discards the result,
will only have the use set $\set{\textsf{logger}}$: $\LET z = x_f\,\textsf{unit}\IN \textsf{unit}$.
In the previous system \cite{DBLP:journals/toplas/BoruchGruszeckiOLLB23},
both terms capture $\set{\textsf{console}, \textsf{logger}}$
though the latter did not use \textsf{console}.
Besides,
the premises and the conclusion share the same use set $C$,
which is more uniform and streamlined.
One can always find a use set that accounts for all the capabilities used by the premises
and use \ruleref{sub} to make this rule applicable.

Beyond providing a theoretical basis for rcaps,
System \capless{} itself is a more principled and expressive foundation for capture tracking
that offers additional theoretical advantages,
which are further discussed in\condtext{}{ our report} \appendixref{B.2}{sec:capless:theoretical-benefits}.
The evaluation rules are almost identical to System \ccformal{},
and we elide them for brevity.
See \appendixref{A.2}{sec:reduction} for the full details.

\section{System \cappy{}: A Surface-Language Calculus for Reach Capabilities}\label{sec:cappy}

System \cappy{} is a core calculus that models the surface language of capture checking with \lstinline|cap|s and reach capabilities.
It is built on the previous System \ccformal{} by Boruch-Gruszecki et al.~\cite{DBLP:journals/toplas/BoruchGruszeckiOLLB23}.
We support the same lightweight end-user notation in types and extend it with \lstinline|@use|
parameters and reach capabilities \lstinline|x*| motivated in \Cref{sec:motivation}.
The semantics of \cappy{} is defined in terms of a type-preserving translation to System \capless{}
with explicit capture quantifiers (\Cref{sec:capless,sec:metatheory}).

\subsection{Syntax}\label{sec:cappy-syntax}
\begin{wide-rules}\noindent
	{\footnotesize\begin{multicols}{3}\noindent
		\begin{flalign*}
			x,\,y,\,z,\,\CAP         \tag*{\textbf{Variable}}\\
			X,\,Y,\,Z                \tag*{\textbf{Type Variable}}\\
			\kappa                \tag*{\textbf{Typedef Name}}\\
			s,\,t,\,u\coloneqq\ &           \tag*{\textbf{Term}}\\
			&a                              \tag*{answer}\\
			&x\,y                              \tag*{application}\\
			&x[S]                              \tag*{type application}\\
			&x[C]                              \tag*{capt. application}\\
			&\LET x = t \IN u                  \tag*{let}\\
			&C\UNBOX x                  \tag*{unbox}\\
			v\coloneqq\ &           \tag*{\textbf{Value}}\\
			&\lambda^\alpha(x: T)t                      \tag*{function}\\
			&\lambda[X<:\top]t                      \tag*{type func.}\\
   \end{flalign*}
	 \begin{flalign*}
			&\lambda[c]t                      \tag*{capture function}\\
			&\BOX x                    \tag*{box}\\
			C,\,D\coloneqq\ &\set{\theta_1,\cdots,\theta_n}     \tag*{\textbf{Capture Set}}\\
			R,\,S\coloneqq\ &           \tag*{\textbf{Shape Type}}\\
			&\top                    \tag*{top}\\
			&X                    \tag*{type variable}\\
			&\forall^\alpha(x: T)U             \tag*{function}\\
			&\forall[X<:\top]T             \tag*{type function}\\
			&\forall[c]T             \tag*{capture function}\\
			&\BOX T             \tag*{boxed}\\
			&\kappa[T_1,\cdots,T_n]  \tag*{applied type}\\
			\alpha\coloneqq\ & \epsilon \mid \USE           \tag*{\textbf{Use Annotation}}\\
			{d}\coloneqq\ & (X_1^{\nu_1},\cdots,X_n^{\nu_n})\mapsto S \tag*{\textbf{Type Definition}}\\
   \end{flalign*}
   \begin{flalign*}
			\theta\coloneqq\ & x \mid {x^*}           \tag*{\textbf{Capture}}\\
			a\coloneqq\ & x \mid v           \tag*{\textbf{Answer}}\\
			{\nu}\coloneqq\ & + \mid -       \tag*{\textbf{Variance}}\\
			T,\,U\coloneqq\ &           \tag*{\textbf{Type}}\\
			&S\capt C                    \tag*{capturing}\\
			&S                    \tag*{pure}\\
			\G,\,\Delta\coloneqq\ &           \tag*{\textbf{Context}}\\
			&\emptyset                    \tag*{empty}\\
			&\G, x: T                    \tag*{term binding}\\
			&\G, X<:\top                    \tag*{type binding}\\
			&\G, c                    \tag*{capture binding}\\
			{\cctx}\coloneqq\ &           \tag*{\textbf{Typedef Context}}\\
			&\emptyset                    \tag*{empty}\\
			&\cctx, \kappa=d                    \tag*{typedef}\\
	\end{flalign*}
	\end{multicols}}
	\vspace{-3em}
	\caption{Abstract syntax of System \cappy{}.}
	\label{fig:cappy-syntax}
\end{wide-rules} 
Figure~\ref{fig:cappy-syntax} shows the syntax of System \cappy{}.
The syntax is close to that of System \ccformal{} \cite{DBLP:journals/toplas/BoruchGruszeckiOLLB23},
basically System F$_\leq$ with captures and boxing.
The main difference is the addition of reach capabilities $x^*$, the use-annotation $\alpha$,
and applied types $\kappa[T_1,\cdots,T_n]$.
Following System \capless{} and \ccformal{}, 
terms are in monadic normal form (MNF) \cite{DBLP:conf/popl/HatcliffD94}, i.e., arguments to operations are always
let-bound variables. 
In the formal syntax, the presence/absence of \lstinline|@use| annotations on parameters is encoded by
annotations $\alpha$, e.g., \lstinline|(@use x: box (File^C)) ->{D} Int|
becomes $(\forall^\USE (x: \BOX (\mathsf{File} \capt C))\mathsf{Int})\capt D$.

Type abstractions take pure type arguments, i.e., type parameters are not qualified with captures
and qualified type arguments must be put into boxes. Boxing plays a crucial role in
enforcing the scope-safety of capabilities \cite{DBLP:journals/toplas/BoruchGruszeckiOLLB23}.
Type parameters are also unbounded (i.e., always being bounded by $\top$). This
choice arises from the translation of System \cappy{} into \capless{} (cf.~\Cref{sec:translation-to-capless}).

An applied type $\kappa[T_1,\cdots,T_n]$ applies a type definition $\kappa$ to the type arguments.
A global type definition context $\cctx$ that contains a list of type definitions is assumed.
In other words, System \cappy{} is parameterized by a type definition context $\cctx$.
A type definition $\kappa = (X_1^{\nu_1},\cdots,X_n^{\nu_n})\mapsto S$ acts like a type macro that,
given a list of type arguments $T_1,\cdots,T_n$,
expands into the type $[X_1:=T_1,\cdots,X_n:=T_n]S$.
$\nu_i$ is the variance of the $i$-th type argument,
which can be either $+$ (covariant) or $-$ (contravariant).
Type definitions enable control over reach-capability expansion (cf.~\Cref{sec:motiv:typedefs}) and
patterns like church-encoded data types in System \cappy{}.
Their interaction with reach refinement will be discussed in \Cref{sec:cappy:type-definitions}.

\subsection{Type System}\label{sec:cappy-statics}\label{sec:cappy-typing}

\begin{figure*}[htbp]
\footnotesize

\flushleft{\textbf{Typing \quad $\typ{C}{\G}{t}{T}$}}

\vspace{-1em}

\begin{multicols}{4}

\infrule[\rruledef{var}]
{x: S\capt C\in \G\\\RefineReach{\set{x^*}}{S}{S'}}
{\typ{\set{x}}{\G}{x}{S'\capt\set{x}}}

\infrule[\rruledef{sub}]
{\typ{C'}{\G}{a}{T'}\\
 \subs{\G}{C'}{C}\andalso
 \subs{\G}{T'}{T}\\
 \wf{\G}{C, T}}
{\typ{C}{\G}{a}{T}}

\infrule[\rruledef{box}]
{\\ \typ{C'}{\G}{x}{S\capt C}}
{\typ{\set{}}{\G}{\BOX x}{\BOX (S\capt C)}}

\infrule[\rruledef{unbox}]
{\typ{C}{\G}{x}{\BOX (S\capt C)}}
{\typ{C}{\G}{C\UNBOX x}{S\capt C}}

\end{multicols}

\begin{multicols}{2}

\infrule[\rruledef{abs}]
{\typ{C}{\G, x: T}{t}{U}\andalso
 \wf{\G}{T}\andalso
 x^*\notin C\ \text{if $\alpha = \epsilon$}}
{\typ{\set{}}{\G}{\lambda^\alpha (x: T)t}{(\forall^\alpha(x: T) U)\capt \left(C\setminus \set{x,x^*}\right)}}

\infrule[\rruledef{app}]
{\typ{C'}{\G}{x}{(\forall^\alpha(z: T) U)\capt C}\andalso
 \typ{C'}{\G}{y}{S\capt D}\\
 \subs{\G}{S\capt D}{T}\andalso
 \subs{\G}{\dcs{\G}{S}}{C'}\ \text{if $\alpha = \USE$}}
{\typ{C'}{\G}{x\,y}{[z^*:=_{+}\dcs{\G}{S}][z:=y]U}}

\end{multicols}

\begin{multicols}{2}

\infrule[\rruledef{cabs}]
{\typ{C}{\G,c}{\lambda[c]t}{T}\andalso\wf{\G}{C}}
{\typ{\set{}}{\G}{\lambda[c]t}{(\forall[c]T)\capt C}}

\infrule[\rruledef{capp}]
{\typ{C}{\G}{x}{(\forall[c]T)\capt C'}\andalso\wf{\G}{D}}
{\typ{C}{\G}{x[D]}{[c:=D]T}}
  
\end{multicols}

\begin{multicols}{3}

\infrule[\rruledef{tabs}]
{\typ{C}{\G, X<:\top}{t}{T}}
{\typ{\set{}}{\G}{\lambda[X<:\top]t}{(\forall[X<:\top] T)\capt C}}

\infrule[\rruledef{tapp}]
{\typ{C}{\G}{x}{(\forall[X<:\top] T)\capt C'}\\ \CAP\notin\dcs{\G}{S}}
{\typ{C}{\G}{x[S]}{[X:=S]T}}

\infrule[\rruledef{let}]
{\typ{C}{\G}{t}{T}\andalso
 \typ{C}{(\G, x: T)}{u}{U}\\
 \wf{\G}{C,U}}
{\typ{C}{\G}{\LET x = t\IN u}{U}}

\end{multicols}

\vspace{-2em}

\flushleft{\textbf{Subcapturing \quad $\subs{\G}{C_1}{C_2}$}\ \text{same as \Cref{fig:all-typing} but without the \rruleref{sc-bound} rule}}

\flushleft{\textbf{Subtyping \quad $\subs{\G}{T_1}{T_2}$}}

The \ruleref{top} and \ruleref{capt} rules are the same as in Figure~\ref{fig:all-typing}.
The \ruleref{trans} and \ruleref{refl} rules are the same as in Figure~\ref{fig:all-typing} but work on capturing types.

\vspace{-0.5em}

\begin{multicols}{2}

\infrule[\rruledef{boxed}]
{\subs{\G}{T_1}{T_2}}
{\subs{\G}{\BOX T_1}{\BOX T_2}}

\infrule[\rruledef{cfun}]
{\subs{(\G,c)}{T_1}{T_2}}
{\subs{\G}{\forall[c]T_1}{\forall[c]T_2}}

\end{multicols}

\begin{multicols}{2}

\infrule[\rruledef{fun}]
{\subs{\G, x: T_2}{U_1}{U_2}\andalso
 \subs{\G}{T_2}{T_1}\andalso
 \alpha_1\preceq \alpha_2}
{\subs{\G}{\forall^{\alpha_1}(x: T_1) U_1}{\forall^{\alpha_2}(x: T_2) U_2}}

\infrule[\rruledef{tfun}]
{\subs{\G, X<:\top}{T_1}{T_2}}
{\subs{\G}{\forall[X<:\top] T_1}{\forall[X<:\top] T_2}}

\end{multicols}

\begin{multicols}{2}

\infrule[\rruledef{applied-p}]
{\kappa=(X_1^{\nu_1},\cdots, X_i^+,\cdots, X_n^{\nu_n})\mapsto S\in\cctx\andalso
 \subs{\G}{T_i}{T'_i}}
{\subs{\G}{\kappa[T_1,\cdots,T_i,\cdots,T_n]}{\kappa[T_1,\cdots,T'_i,\cdots,T_n]}}

\infrule[\rruledef{applied-m}]
{\kappa=(X_1^{\nu_1},\cdots, X_i^-,\cdots, X_n^{\nu_n})\mapsto S\in\cctx\andalso
 \subs{\G}{T'_i}{T_i}}
{\subs{\G}{\kappa[T_1,\cdots,T_i,\cdots,T_n]}{\kappa[T_1,\cdots,T'_i,\cdots,T_n]}}

\end{multicols}

\infrule[\rruledef{dealias}]
{\kappa=(X_1^{\nu_1},\cdots,X_n^{\nu_n})\mapsto S\in\cctx\andalso
 \forall X_i^+,\CAP\notin\dcs{\G}{T_i}}
{\subs{\G}{\kappa[T_1,\cdots,T_n]}{[X_1:=T_1,\cdots,X_n:=T_n]S}\andalso
 \subs{\G}{[X_1:=T_1,\cdots,X_n:=T_n]S}{\kappa[T_1,\cdots,T_n]}}

\flushleft{\textbf{Reach Refinement \quad $\RefineReach{C}{T}{U}$}}

\vspace{-1em}

\begin{multicols}{4}

\infax[\rruledef{r-top}]
{\RefineReach{C}{\top}{\top}}

\infax[\rruledef{r-tvar}]
{\RefineReach{C}{X}{X}}

\infrule[\rruledef{r-capt}]
{\RefineReach{D}{S}{S'}}
{\RefineReach{D}{S\capt C}{S'\capt [\CAP:=D]C}}

\infrule[\rruledef{r-cfun}]
{\RefineReach{D}{T}{T'}}
{\RefineReach{D}{\forall[c]T}{\forall[c]T'}}

\infrule[\rruledef{r-boxed}]
{\RefineReach{D}{T}{T'}}
{\RefineReach{D}{\BOX T}{\BOX T'}}

\end{multicols}
\begin{multicols}{3}

\infrule[\rruledef{r-tfun}]
{\RefineReach{D}{T}{T'}}
{\RefineReach{D}{\forall[X<:\top]T}{\forall[X<:\top]T'}}

\infax[\rruledef{r-fun}]
{\\
 \RefineReach{D}{\forall^\alpha(z:T)U}{\forall^\alpha(z:T)U}}

\infrule[\rruledef{r-applied}]
{\kappa=(X_1^{\nu_1},\cdots,X_n^{\nu_n})\mapsto S\in\cctx\\
 \forall X_i^+,\RefineReach{D}{T_i}{T'_i}\andalso
 \forall X_i^-,T'_i=T_i}
{\RefineReach{D}{\kappa[T_1,\cdots,T_n]}{\kappa[T'_1,\cdots,T'_n]}}

\end{multicols}

\vspace{-1.5em}
\caption{Static semantics of System \capcalculus{}.}\label{fig:cappy-all-typing}
\end{figure*}

The typing judgement $\typ{C}{\G}{t}{T}$ in \Cref{fig:cappy-all-typing} %
is formulated with use sets
like System \capless{} (\Cref{sec:capless}).
Rule \rruleref{var} mostly matches the one in System \capless{}.
The additional
\emph{reach refinement} (cf.\ \Cref{sec:motiv:reach-capabilities}) of the variable's assumed type
$S$ to $S'$, which replaces certain occurrences of $\CAP$ in $S$ with the reach capability
$x^*$.
We explain the details of reach refinement later in \Cref{sec:reach-refinement}.

Rules governing box introduction \rruleref{box} and box elimination \rruleref{unbox} follow
\ccformal{}. Boxes are considered to be pure values, thus they have an empty use set, and
the unboxing operation is only allowed in a context where the use set matches the box's capabilities, as before.

The introduction and elimination of type abstractions \rruleref{tabs} and \rruleref{tapp} are standard.
$\CAP$ is not allowed in the deep capture set of the type argument $S'$ in \rruleref{tapp},
as it leads to ambiguity in the meaning of $\CAP$s
and breaks the translation from \cappy{} to \capless{} (cf. \Cref{sec:translation-to-capless}).

Let bindings \rruleref{let} are typed in a standard manner.
The mention of locally bound variable $x$
has to be avoided in the use set $C$ and the result type $U$
by subcapturing and subtyping, respectively.
We discuss the more
interesting typing rules for abstraction and application next.

\subsubsection{Abstraction and Application}\label{sec:cappy-fun-types}

The \rruleref{abs} rule mostly follows that of System \capless{}.
A function that does not declare its
parameter $x$ as used ($\alpha = \epsilon$) is accordingly barred from using the associated reach
capability $x^*$ in the use set $C$ of the body.

Dependent function application \rruleref{app} resolves reach capabilities. The reach capability of a function with use-parameter ($\alpha = \USE$) is
substituted with the call-site argument's \emph{deep capture set}:

\begin{definition}[Deep Capture Set]
The deep capture set of a type $T$ under context $\G$,
denoted as $\dcs{\G}{T}$,
is defined as follows:
$$\small
\begin{array}{r@{\hspace{5pt}}c@{\hspace{5pt}}l@{\hspace{20pt}}r@{\hspace{5pt}}c@{\hspace{5pt}}l}
\dcs{\G}{\top} &=& \set{} & \dcs{\G}{\forall[X<:S]T} &=& \dcs{\G}{T}\\
\dcs{\G}{X} &=& \dcs{\G}{S}\quad\text{if $X<:S\in\G$} & \dcs{\G}{\BOX T} &=& \dcs{\G}{T}\\
\dcs{\G}{\forall^\alpha (z: T)U} &=& \dcs{\G}{U}\setminus\set{z,z^*} & \dcs{\G}{S\capt C} &=& \dcs{\G}{S}\cup C\\
& & & \dcs{\G}{\forall[c]T} &=& \dcs{\G}{T}\setminus\set{c}\\
\dcs{\G}{\kappa[T_1,\cdots,T_n]} &=& \bigcup_{X_i^+}\dcs{\G}{T_i}&\hspace{-5em}\text{ if }\kappa = (X_1^{\nu_1},\cdots,X_n^{\nu_n})\mapsto S\in\cctx& & \\
\end{array}
$$
\end{definition}

The deep capture set notion $\dcs{\G}{T}$ formally answers "what's in the box?" of a boxed capture type. %
Our definition generalizes to a richer language like Scala with generic type constructors, e.g., $\dcs{\G}{\mathsf{List}[S]} = \dcs{\G}{S}$.
Variance-aware substitution $[x^*:=_{\nu} C]$
replaces covariant occurrences of $x^*$ with $C$
and contravariant ones with $\set{}$.
$\nu$ can be either $+$ (covariant) or $-$ (contravariant).
The starting variance is $+$. 
This ensures the subtyping relation is preserved under substitution.
 
\subsubsection{Subcapturing and Subtyping}\label{sec:cappy-subtyping}

The subsumption rule \rruleref{sub} allows refining the use set $C$ to $C'$
and the type $T$ to $T'$.
It can only be applied on an answer $a$, which is either a value or a variable
due to the translation mechanism from \cappy{} to \capless{},
which will be discussed in \Cref{sec:translation-to-capless}.

Like System \capless{}, subcapturing and subtyping follows those of System \ccformal{} \cite{DBLP:journals/toplas/BoruchGruszeckiOLLB23},
which are mostly unsurprising. 
We define an ordering $\preceq$ on annotations defined by $\alpha\preceq\alpha$ and $\epsilon\preceq\USE$, i.e.,
non-use functions can pass for use functions, but not vice versa (as in \rruleref{fun}).
The rules \rruleref{applied-p} and \rruleref{applied-m}
supports argument subtyping for applied types with respect to the variance.
\rruleref{dealias} deals with the expansion of applied types.

\subsubsection{Reach Refinement}\label{sec:reach-refinement}
Reach refinement $\RefineReach{C}{T}{U}$ replaces certain covariant occurrences of $\CAP$ in type $T$
with the capture set $C$, where $C$ is often the reach capability of a variable (e.g., in the \rruleref{var} rule). 
As shown in Figure~\ref{fig:cappy-all-typing}, refinement is defined recursively on the structure of types.
The base cases for top types and type variables \rruleref{r-top} and \rruleref{r-tvar} leave the type unchanged.
For capturing types \rruleref{r-capt}, refinement recursively applies to the shape type and directly substitutes occurrences of $\CAP$ in the capture set with $C$.
For boxed types \rruleref{r-boxed} and type functions \rruleref{r-tfun}, refinement is applied to the inner type.
For applied types \rruleref{r-applied}, refinement is applied to covariant type arguments
and contravariant ones are left unchanged.
Reach refinement touches neither the domain nor the codomain of function types \rruleref{r-fun},
which is due to the translation scheme discussed in \Cref{sec:motiv:reach-capabilities}.
This is further discussed in\condtext{}{ the technical report} \appendixref{B.3}{sec:reach-refinement-function-types}.

\subsubsection{Type Definitions}\label{sec:cappy:type-definitions}
Admittedly, this restriction that prevents refining function types indeed imposes a loss in expressiveness.
For example, church-encoded data types cannot be properly refined.
Nevertheless,
we can recover the expressiveness thanks to type definitions and applied types.
As an example, consider the church-encoded pair type:
{\footnotesize
\[
\textsf{Pair} = (X_1^+, X_2^+)\mapsto \forall[X_R<:\top]\forall({z}: X_1\Rightarrow X_2\Rightarrow X_R) X_R
\]}

\noindent Assume that the applied type $\textsf{Pair}[\BOX \textsf{IO}\capt\set{\CAP}, \BOX \textsf{IO}\capt\set{\CAP}]$
is bound to a variable $x$.
\rruleref{r-applied} refines this type to $\textsf{Pair}[\BOX \textsf{IO}\capt\set{x^*}, \BOX \textsf{IO}\capt\set{x^*}]$,
which can then be expanded to
$\forall[X_R<:\top]\forall({z}: \BOX \textsf{IO}\capt\set{x^*}\Rightarrow \BOX \textsf{IO}\capt\set{x^*}\Rightarrow X_R) X_R$
by the \rruleref{dealias} rule in subtyping.
Essentially, applied types change the way the existential quantifications are scoped:
In particular, the applied type
$\textsf{Pair}[\BOX \textsf{IO}\capt\set{\CAP}, \BOX \textsf{IO}\capt\set{\CAP}]$
is interpreted as 
$\EXCAP{c}\textsf{Pair}[\BOX \textsf{IO}\capt\set{c}, \BOX \textsf{IO}\capt\set{c}]$.
The existential variable is quantified at the outer-level.
The $\CAP$s can thus safely be replaced by $x^*$ in the refinement.

\subsection{Translation to System \capless{}}
\label{sec:translation-to-capless}

System \cappy{} is a surface language that ``desugars'' to \capless{}
in terms of a type-preserving translation (\Cref{sec:type-preserving-translation}).
The core idea of the translation is to
recursively convert
the occurrences of $\CAP$ in parameter positions
to universal capture parameters,
and those in function result types to existential quantifications.
Consider the System \cappy{} type of the \textsf{mkIterator} function (\Cref{sec:motiv:reach-capabilities}):
{\footnotesize
\[
\forall[X<:\top]\forall^\USE(x: \textsf{List}[\BOX (\textsf{Unit}\Rightarrow X)])\textsf{Iterator}[X]\capt\set{x^*}
\]
}
It translates to:
{\footnotesize
\[
\forall[X<:\top]\forall[c_x<:\set{}]\forall[c_{x^*}](\forall(x: \textsf{List}[\BOX (\textsf{Unit}\rightarrow X)\capt\set{c_{x^*}}]\capt\set{c_x})\textsf{Iterator}[X]\capt\set{c_{x^*}})\capt\set{c_{x^*}}
\]
}

\noindent
Here, two universal capture parameters $c_x$ and $c_{x^*}$ are introduced.
$c_x$ corresponds to the outermost capture set of the parameter $x$,
and is upper-bounded by an empty set in this case.
$c_{x^*}$ corresponds to the $\CAP$s inside the box of $x$,
which is exactly the \emph{meaning} of the reach capability $x^*$.
$x$ is a use-parameter, thus $c_{x^*}$ is allowed to be used in the body of the function.

One surprising aspect of the translation is that
subtyping in the source language
induces term transformations in the target language.
Consider a \cappy{} term $a$ whose type $T$ is a subtype of $U$.
Let $T'$ and $U'$ be the translated types of $T$ and $U$ respectively,
and $a'$ be the translated term of $a$ of type $T'$.
In general, $a'$ does not directly conform to the type $U'$.
It needs to be transformed to a term $a''$ to conform to the type $U'$.
In fact,
widening a capture set $C$ to $\set{\CAP}$ in the source language
corresponds to packing the capture set into an existential in the target language.
For instance, consider the types $\forall(x: \textsf{File}\capt\set{\CAP}) \textsf{File}\capt\set{x}$ and 
$\forall(x: \textsf{File}\capt\set{\CAP})\textsf{File}\capt\set{\CAP}$. The first type is a subtype of the second in \cappy{}.
However, their translations differ:
{\footnotesize
\begin{align*}
&\forall(x: \textsf{File}\capt\set{\CAP}) \textsf{File}\capt\set{x} && \text{translates to} && \forall[c_{x}]\forall(x: \textsf{File}\capt\set{c_x}) \textsf{File}\capt\set{c_x} \\
&\forall(x: \textsf{File}\capt\set{\CAP})\textsf{File}\capt\set{\CAP} && \text{translates to} && \forall[c_{x}]\forall(x: \textsf{File}\capt\set{c_x}) \EXCAP{c}\textsf{File}\capt\set{c}
\end{align*}
}

\noindent To transform a term of the first translated type to one of the second type, we must perform an eta-expansion to pack the capture set $\set{c_x}$ into the existential quantification:
{\footnotesize
\[
\lambda[c_{x}]\lambda(x: \textsf{File}\capt\set{c_x}).\<\set{c_x}, a_f\,x\>
\]
}
where $a_f$ is the translated version of the original function.

The necessity of term transformations motivates us to restrict subtyping on answers \rruleref{sub},
which ease the transformation process;
and to restrict type parameters to be unbounded,
due to the lack of term-level witnesses for subtyping over type parameters.
When translating subtyping derivations to System \capless{}, 
we require term-level witnesses that can be transformed to adapt between different translated types.
However, subtyping between bounded type parameters exists purely at the type level 
with no corresponding term representation that could be transformed during translation.
For instance, consider a type function $f: \forall[X <: (\textsf{IO}\capt\set{\textsf{io}} \rightarrow \textsf{Unit})] \ldots$
and an application $f[\textsf{IO}\capt\set{\CAP} \rightarrow \textsf{Unit}]$,
where $\textsf{IO}\capt\set{\CAP} \rightarrow \textsf{Unit} <: \textsf{IO}\capt\set{\textsf{io}} \rightarrow \textsf{Unit}$.
Translating this application form requires the source-language subtyping relation which materializes
into term transformation in the target language,
but no term is available to be adapted at the type-application site.
Since we cannot produce the necessary term-level adaptations for bounded type parameters,
we restrict all type parameters to be unbounded, so that the translation remains complete.
\section{Metatheory}\label{sec:metatheory}

We prove the type soundness of System \capless{} (\Cref{sec:capless}),
through standard progress and preservation theorems as well as its scope safey.
Those proofs are fully mechanized in Lean 4.

In addition,
we relate the surface-language System \cappy{} (\Cref{sec:cappy})
to the core language System \capless{}
via a type-preserving translation.
Specifically, any well-typed program in \cappy{} can be translated
to a well-typed one in \capless{} with an equivalent type.

\subsection{Type Soundness}

We take a standard syntactic approach towards type soundness, proving the following theorems.

\begin{theorem}[Preservation]
\label{theorem:preservation}
If
(1) ${\typSta{\sta}{\G}}$,
(2) $\typ{C}{\G}{t}{E}$,
and (3) $\<\sta\BAR t\>\red\<\sta'\BAR t'\>$,
then there exist $C'$ and $\Delta$ such that
(1) ${\typSta{\sta'}{(\G,\Delta)}}$;
and (2) $\typ{C'}{(\G,\Delta)}{t'}{E}$.
\end{theorem}
Here, ${\typSta{\sta}{\G}}$ denotes \emph{store typing}:
given a store $\sta$,
this relation produces a typing context $\G$
whose bindings assign a type to each corresponding value in the store
(defined in\condtext{}{ the technical report} \appendixfigref{11}{fig:store-typing}).
$\<\sta\BAR t\>\red\<\sta'\BAR t'\>$ denotes the reduction relation.

\begin{theorem}[Progress]
\label{theorem:progress}
Given (1) ${\typSta{\sta}{\G}}$, and (2) $\typ{C}{\G}{t}{E}$,
we can show that either $t$ is an answer,
or there exist $\sta'$ and $t'$ such that
$\<\sta\BAR t\>\red\<\sta'\BAR t'\>$.
\end{theorem}

\subsection{Scope Safety}

Following \citeauthor{DBLP:journals/toplas/BoruchGruszeckiOLLB23},
we add an extension to System \capless{} which introduces scoped capabilities.
By proving it sound, we show that the system can ensure the scoping of capabilities.

\begin{figure*}[tbp]
\footnotesize

\begin{minipage}{.5\textwidth}
\flushleft{\textbf{Syntax}}
\begin{align*}
  s,\,t,\,u\coloneqq\ & \BOUNDARY[S] \AS \langle c,x\rangle \IN t \mid \cdots           \tag*{\textbf{Term}}\\
  R,\,S\coloneqq\ &\BREAK[S] \mid \cdots           \tag*{\textbf{Shape Type}}\\
\end{align*}
\end{minipage}%
\begin{minipage}{.5\textwidth}
\flushleft{\quad\textbf{Subtyping} \quad $\subs{\G}{S_1}{S_2}$}
\infrule[\ruledef{break}]
{\subs{\G}{S_2}{S_1}}
{\subs{\G}{\BREAK[S_1]}{\BREAK[S_2]}}
\end{minipage}

\flushleft{\textbf{Typing} \quad $\typs{\G}{t}{T}$}
\begin{multicols}{1}
\infrule[\ruledef{boundary}]
{\typ{C}{(\G, c:\CAPK, x:\BREAK[S]\capt\set{c})}{t}{S}\andalso
 \wf{\G}{S}}
{\typ{(C\setminus\set{c,x})}{\G}{\BOUNDARY[S]\AS\PACK{c}{x}\IN t}{S}}

\infrule[\ruledef{invoke}]
{\typ{C'}{\G}{x}{\BREAK[S]\capt C}\andalso\typ{C'}{\G}{y}{S}}
{\typ{C'}{\G}{x\,y}{E}}
\end{multicols}

\vspace{-1.5em}
\caption{Extensions to static rules of \calculus{}.}\label{fig:static-extensions}

\end{figure*} 
\subsubsection{Static Semantics of Scoped Capabilities}
Figure~\ref{fig:static-extensions} presents the extensions. We add a control delimiter
$\BOUNDARY[S]\AS\langle c,x\rangle\IN t$ (mirroring \lstinline|boundary|/\lstinline|Break| in
Scala~3) that introduces a scope with a fresh abstract capture $c$ and a scoped capability
$\BREAK[S]$. The capability may be used only within its defining boundary; the type system enforces
this as follows: \ruleref{boundary} checks the body under $\BREAK$, binds the fresh $c$, and
requires the result type $S$ to be well-formed in the outer context $\Gamma$, which suffices to
enforce scoping. \ruleref{invoke} types invocations of $\BREAK$, and \ruleref{break} provides
subtyping between $\BREAK$ capabilities.

\subsubsection{Dynamic Semantics of Scoped Capabilities}
In the reduction semantics we introduce (1) runtime labels to identify boundaries and (2) a scoping
construct that delimits a boundary together with its $\BREAK$ capability. Invoking $\BREAK$ requires
a matching scope in the surrounding evaluation context; otherwise evaluation is stuck. 
For brevity, we elide the full runtime forms here and refer to\condtext{}{ the technical report} \appendixref{A.3}{sec:dynamic-rules-scopedcaps} for the complete definition.

\subsubsection{Type Soundness}
We prove standard progress and preservation theorems
to establish the type soundness of the extended system.
Its type soundness implies that the system enforces the scoping discipline of the $\BREAK$ capability,
as a scope extrusion will lead to a stuck evaluation.
The proof is fully mechanized in Lean 4.

\subsection{Type Preserving Translation}
\label{sec:type-preserving-translation}

We develop a translation system that
translates well-typed programs in System \cappy{} to System \capless{}.
This essentially provides an understanding of $\CAP$s and reach capabilities:
they can be understood as existential and universal quantifications of capture sets.
The translation is stated and proven with pen and paper.
See\condtext{}{ our technical report} \appendixref{D}{sec:proof} for the full details.

\begin{restatable}[Translation Preserves Typing]{theorem}{thmtranslation}
\label{theorem:typing-interp}
Let $\tau = {\<D,\rho,\rho^*\>}$ be a proper translation context under type contexts $\G$ and $\Delta$, 
and $\typ{C}{\G}{t}{T}$ be a typing derivation in System \cappy{}.
Then, if $t$ is either of the application form $x\,y$
or the let-binding form $\LET z = s\IN u$,
there exists a term $t'$ in System \capless{} such that
$\typ{\embed{C}^{D'}}{\Delta}{t'}{\EXCAP{c}\embed{T}^{\set{c}}}$ for some $D'$;
otherwise,
there exists a term $t'$ in System \capless{} such that
$\typ{\embed{C}^{D'}}{\Delta}{t'}{\embed{T}^{D'}}$ for some $D'$,
and $t'$ is an answer when $t$ is an answer.
\end{restatable}
${\<D,\rho,\rho^*\>}$ is the translation context in System \decap{},
a translation system that converts capture sets and types in \cappy{}
into those in \capless{}.
The translation context consists of
(1) a capture set $D$ assigning meaning to $\CAP$s,
(2) and two mappings $\rho$ and $\rho^*$ that map capabilities (including reach capabilities) to their corresponding capture sets.
The notation ${\embed{T}^{D}}$ denotes the translation of a \cappy{} type $T$ to a \capless{} type,
and similarly for $\embed{C}^{D}$, which denotes the translation of a \cappy{} capture set $C$.
The key idea of the translation is to
interpret $\CAP$s in the argument types as universal quantifications
and those occuring covariantly in the result type as existential quantifications.
\section{Capture Tracking for Asynchronous Programming: A Case Study}
\label{sec:case_study}

We study the usage of capture checking and more specifically reach capabilities through
some real-world examples of asynchronous programming.

\subsubsection*{Futures and Scoping}

Scala supports a simple form of concurrency with Futures.
A \lstinline|Future| takes a computation and runs it concurrently under a provided execution context.
Note that despite the computation not being stored by the Future,
it can hold on to capabilities captured by the computation after \lstinline|Future.apply|
returns. Therefore, a Future should capture both the execution context and its computation.
We model the Future constructor
using a context function~\cite{DBLP:journals/pacmpl/OderskyBLBMS18}:
\begin{lstlisting}
object Future { def apply[T](comp: => T)(using ec: ExecutionContext^): Future[T]^{ec, comp} }
\end{lstlisting}
The following program attempts to read a file and filter its content, using
the standard \smallcode{Using} resource-pattern.
However, notice that \smallcode{Future.apply}
does not run the reading procedure to the end, instead returning immediately and causing \smallcode{Using} to
close the file right after.

\begin{lstlisting}[columns=fixed]
def findLines(p: String)(using ec: ExecutionContext^): Future[Seq[String]]^{ec} =
  Using(File.open("text.txt")): (file: File^) => // ERROR: file leaked
    Future.apply:                // : Future[Seq[String]]^{ec, file}
      file.readLines()           // : Iterator[String]^{file}
          .filter(_.contains(p)) // : Iterator[String]^{file}
          .toSeq                 // : Seq[String]
\end{lstlisting}
The compiler rejects it, detecting that \smallcode{Future.apply} returns a \lstinline|Future| capturing \smallcode{file} which escapes the \smallcode{Using} scope.
To fix this, the file has to be opened \emph{under} the \smallcode{Future.apply} call.

\subsubsection*{Composing Futures}
When writing concurrency code, it is common to create multiple concurrent computations, and then
combine them either by requiring both or either (racing) futures:
\begin{lstlisting}
extension [T](@use xs: Seq[Future[T]^])
  def all(using ec: ExecutionContext^): Future[Seq[T]]^{ec, xs*}
  def first(using ec: ExecutionContext^): Future[T]^{ec, xs*}
\end{lstlisting}
The extension methods \lstinline|all| and \lstinline|first| create a future that combines or races the
input futures. The resulting \lstinline|Future| upholds all scoping restrictions of the futures in the input sequence
by capturing the reach capability of the sequence itself.
As we saw in Section \ref{sec:motiv:reach-capabilities}, \lstinline|xs| is pure,
but we want to ``open'' the sequence of futures to await for their results, effectively using the futures' captures.
For this reason, the capture checker requires the \lstinline|@use| annotation; and when
added, it will propagate the reach capability \lstinline|xs*| to the capture set of the caller.
This is demonstrated in the following example:
\begin{lstlisting}
withIO:
  withThrow:
    val futs: Seq[Future[Int]^{async, io, throw}] =
      val f1: Future[Int]^{async, io} = Future(useIO())
      val f2: Future[Int]^{async, throw} = Future(useThrow())
    () => futs.all // (() -> Future[Seq[Int]])^{async, io, throw}
\end{lstlisting}
The returned function only captures \lstinline|futs|, a pure variable. However, \lstinline|@use| requires the \lstinline|.all|
call to add the reach capabilities of \lstinline|futs| to the closure, resulting in the rejection of the program.
Without this check, we would be able to invoke \lstinline|io| and \lstinline|throw| outside of their scope.

The implementation of \lstinline|.all| uses a similar construct to the \lstinline|collect| example in Section \ref{sec:introduction},
named \lstinline|Future.Collector[T]|. A collector takes a sequence \lstinline{xs} of \lstinline|Future[T]| and exposes a read-only channel that
passes back the futures as they are completed.
Naturally, the returned futures have the same capture set as the union of all captures of the input futures -- or the reach
capability \lstinline|xs*|.
To implement \lstinline|.all|, we create a \lstinline|Collector| from the given futures,
and chain the results on each as they are completed - guaranteeing that failing futures immediately return the exception
to the caller as soon as possible:
\begin{lstlisting}[columns=fixed]
extension [T](@use xs: Seq[Future[T]^]) def all(using ExecutionContext^) =
  val results = Collector(xs).results // : Channel[Future[T]^{xs*}]
  // create a future that poll the results channel xs.size times
  val poll = (0 until xs.size)
    .foldLeft(Future.unit)((fut, _) => fut.andThen(_ => results.read()))
  poll.andThen: _ => // all futures are resolved at this point
    xs.foldLeft(Future.apply(Seq.empty)): (seq, fut) =>
        seq.zip(fut).map(_ +: _) // : Future[Seq[T]]^{xs*}
\end{lstlisting}

\subsubsection*{Mutable Collectors}\label{sec:mutable-collectors}
A \textit{mutable} version of \lstinline|Collector[T]| is useful for work queues
where jobs dynamically spawn new concurrent tasks,
but results should arrive as soon as possible.
Unfortunately, for mutable collections like \lstinline|MutableCollector[T]|, reach capabilities are not expressive enough:
the collection may be created empty (and hence no longer accept inserting capturing values), but allowing the capture
set of the collection's items to grow is unsound.
In such cases, we reach a middle ground and require the user to \textit{declare in advance} the upper-bound capture set,
with an explicit parameter. The collection then can accept any item conforming to this declared capture set:

\begin{lstlisting}[columns=fixed]
class MutableCollector[T, C^]():
  val results: ReadChannel[Future[T]^{C}]; def remaining: Int
  def add(fut: Future[T]^{C})
def parSearch(run: Node => Seq[Node], start: Node)(using ec: ExecutionContext^) =
  val queue = MutableCollector[List[Node], {ec, run}]()
  queue += Future.apply(Seq(start))
  def loop(): Future[Unit]^{ec, run} =
    if queue.remaining == 0 then Future.unit else // : Future[Unit]^{}
      val nextToProcess = queue.results.read() // : Future[Seq[Node]]^{ec, run}
      nextToProcess.andThen: children =>
        children.foreach(node => queue += Future.apply(run(node)))
        loop()
  loop()
\end{lstlisting}
The above example illustrates the use of mutable collectors to implement parallel tree search.
In \lstinline|parSearch|, we create a mutable collector that accepts \lstinline|Future|s
that can capture the execution context and capabilities used by \lstinline|run|.
We utilize the collector as a work queue to immediately spawn new jobs as soon
as one comes back, resulting in minimal downtime waiting while more children are discovered.
To specify a capture set, we use the \lstinline|{ }| syntax in the type parameter.
Note that there is very little notational overhead when implementing \lstinline|parSearch|:
the mutable collector requires a non-inferrable capture set, and \lstinline|loop|, being recursive,
requires a specified return type. All other type annotations in comments, including capture annotations, are
inferred by the compiler.
\section{Evaluation}\label{sec:evaluation}

While the previous case study demonstrated the practicality of our approach through qualitative analysis of common programming patterns, this section provides a quantitative evaluation by measuring performance and required code changes when adopting capture checking in practice.

\subsection{Implementation}
\label{sec:implementation}

\begin{wraptable}{r}{0.4\textwidth}
	\vspace{-13pt}
	\scriptsize
	\caption{{\scriptsize Compiler performance compiling the capture-checked standard library. Average of 10 runs, measured on Linux PC (Ryzen 3800x, 32GB RAM, NixOS, JDK11)}}\label{tab:compiler-performance}
	\vspace{-10pt}
	\begin{tabular}{lrr}
		\toprule
		\textbf{Phase} & \textbf{Time (ms)}        & \textbf{Throughput}   \\
		\midrule
		Typing         & 10436 $\pm$ 220 (35.80\%) & 4596 $\pm$ 99 loc/s   \\
		CC             & 3896 $\pm$ 48 (13.36\%)   & 12307 $\pm$ 148 loc/s \\
		Other          & 14821 $\pm$ 269 (50.84\%) & 3236 $\pm$ 58 loc/s   \\
		\textbf{Total} & 29154 $\pm$ 395 (100.0\%) & 1645 $\pm$ 22 loc/s   \\
		\bottomrule
	\end{tabular}
	\vspace{-13pt}
\end{wraptable}
The Scala 3 capture checker provides a complete implementation of capture tracking
with reach capabilities, as well as optional universal quantification of capture sets.
It is enabled by an experimental language import.

The capture checker runs after the type checker and several transformation phases.
It re-checks the typed syntax tree of a compilation unit, deriving subcapturing constraints
between capture sets that are solved incrementally with a propagation-based constraint solver that
keeps track of which capabilities are known to form part of a capture set.
To understand the performance impact of capture checking,
we show the compilation time breakdown of compiling the capture-checked standard library in Table~\ref{tab:compiler-performance}.
Capture checking takes approximately half the time compared to the previous typing phase and accounts for less than $15\%$ of the total compilation time,
which is reasonable.
All type arguments are treated as boxed, and boxing and unboxing operations are inferred in an adaptation step
that compares actual against expected types.
 
\subsection{Porting the Scala Standard Library}\label{sub:porting-stdlib}

To evaluate the practicality of capture checking, we have ported a significant portion of the Scala standard library to compile
with the Scala 3 capture checker.
In particular, the collections library is ported in full, alongside common language structures like \lstinline|Array| and
\lstinline|Function| traits.
This shall be referred to as \lstinline|lib-cc|, as opposed to the original standard
library (\lstinline|lib|), which is at version 2.13.15 at the time of comparison.

\newcommand{\totalcolor}{\color[HTML]{6B7A8F}\itshape}

\begin{wraptable}{r}{0.72\textwidth}
	\vspace{-20pt}
	\caption{
		Number of changes to capture-checked standard library, grouped by categories of change.
		"Total in lib-cc" shows the total number of definitions, variables, higher order functions, casts and lines of code in the library.
		"Meaningful" measures lines of changes excluding whole-line comments, blank lines and feature imports.
		For reference, "Total in lib" shows corresponding numbers in the original standard library (Scala 2.13.15),
		including parts not yet implemented by lib-cc, i.e., 66\% of the standard library is collections.
	}
	\label{tab:changes-stdlib}
	\centering
	\scriptsize
	\vspace{-10pt}
	\begin{tabular}{lrrr}
		\toprule
		\textbf{Type of Change}         & \textbf{Added/Changed}                                             & \textbf{Total in lib-cc}                                             & {\totalcolor \textbf{Total in lib}}                                                  \\
		\midrule
		Capture sets on definitions     & \begin{tabular}[c]{@{}r@{}}691 functions\\ 91 classes\end{tabular} & \begin{tabular}[c]{@{}r@{}}6189 functions\\ 684 classes\end{tabular} & {\totalcolor \begin{tabular}[c]{@{}r@{}}11113 functions\\ 1518 classes\end{tabular}} \\
		- Only universal captures       & 351                                                                &                                                                      & {\totalcolor }                                                                       \\
		- Only capture set on returns   & 379                                                                &                                                                      & {\totalcolor }                                                                       \\
		Capture sets on local variables & 53                                                                 & 4583                                                                 & {\totalcolor 6262}                                                                   \\
		Restrict functions to pure      & 52                                                                 & 797                                                                  & {\totalcolor 1203}                                                                   \\
		Reach capabilities              & 19                                                                 &                                                                      & {\totalcolor     }                                                                   \\
		\texttt{@use} annotations       & 12                                                                 &                                                                      &                                                                                      \\
		Unsafe casts                    & 5                                                                  & 810                                                                  & {\totalcolor 1325}                                                                   \\
		Unsafe capture set removal      & 25                                                                 &                                                                      & {\totalcolor }                                                                       \\
		Class hierarchy changes         & 2                                                                  &                                                                      & {\totalcolor }                                                                       \\
		- New classes                   & 3                                                                  &                                                                      & {\totalcolor }                                                                       \\
		\midrule
		\textbf{Total lines}            & \begin{tabular}[c]{@{}r@{}}+1418/-1407 meaningful\end{tabular}     & \begin{tabular}[c]{@{}r@{}}52160 total\\ 31395 code\end{tabular}     & {\totalcolor \begin{tabular}[c]{@{}r@{}}94635 total\\ 48210 code\end{tabular}}       \\
		\bottomrule
	\end{tabular}
	\vspace{-12pt}
\end{wraptable}
We measure the number of changes by direct comparison (using standard \lstinline|diff|) between \lstinline|lib-cc| and \lstinline|lib|
source code. The changes are categorized into Table \ref{tab:changes-stdlib}.
Overall, capture checking annotations require only about 3\% of lines of code to be changed, with over half of the changes
involving only adding capture sets to function and class declarations.
Despite the vast majority (>98.5\%) of \lstinline|lib-cc| being generic collections, about 87\% of all classes
and 88\% of methods do not require any signature change - all existing code will compile as-is.
This is thanks to boxing type parameters as well as built-in type and capture set inference.

\subsubsection{Additional Annotations}

In a lot of cases, the existing signatures work as it is.
Examples are higher order functions like \smallcode{List.map(f: A => B)},
whose parameters are only used and not captured in the return value. 
Scala's capture checking implementation treats ``fat-arrow'' function types as effect-polymorphic,
and therefore \smallcode{List.map} does not require any changes.
There are some false-positives, where a pure function is expected: the
\lstinline|maxBy| method on an \lstinline|IterableOnce[A]| expects a function that extracts a property of \lstinline|A|, and should not perform any effects.
We restrict functions to pure (\lstinline{A -> B}) in such cases.
For many other cases, the existing implementation is sufficient to handle arbitrarily capturing inputs and outputs,
but extra capture set annotations are needed to communicate this fact.
One of such cases involves passing parameters of a possibly capturing type, such as an
iterator in the following method of \lstinline|List[A]|:
\begin{lstlisting}[aboveskip=2pt,belowskip=2pt]
def prependAll(it: IterableOnce[A]^): List[A]
\end{lstlisting}
In this case, \smallcode{prependAll} should consume items from the given iterator
(possibly invoking capabilities captured in it), and then return a new prepended list without capturing \lstinline|it|.
As-is, \smallcode{prependAll} would only allow pure iterators to be passed in,
greatly reducing its usefulness.
However, adding the universal capture annotation completely fixes this issue.
Over half (51.7\%) of all changed signatures involve only this addition of \lstinline|cap|.

In some cases, the return value captures one or more of its parameters (and possibly the current \lstinline|this|).
Of particular note is the implementation of functional operations of iterators and other lazy data structures, e.g.,
the following signature of \lstinline|map| on an \smallcode{Iterator[A]}:
\begin{lstlisting}[aboveskip=2pt,belowskip=2pt]
def map[B](f: A => B): Iterator[B]^{this, f}
\end{lstlisting}
Note that the return type captures both the original iterator (\lstinline|this|) and the mapping function; and
this is also the only form of additional annotation required. Such changes are very common, and make up of
most of the remaining half of changes on signatures.

As expected, all instances of reach captures belong to \lstinline|flatMap| (and \lstinline|flatten|,
which is a special case of \lstinline|flatMap|) implementations of (possibly-)lazy data structures and interfaces.
All such instances in parameter position show up with \lstinline|@use| annotations.
Both can be removed when overriding such interfaces in a strict data structure like \lstinline|Map|.

\subsubsection{Guaranteeing Compatibility}

To enable incremental adoption, Scala allows mixing capture-checked and unchecked modules. However,
compiling them together still presents challenges. One current limitation is that explicit
capture-set parameters cannot yet be used from unchecked modules. For example, the class
\lstinline|ConcatIterator[A]| internally maintains a mutable linked list:
\begin{lstlisting}[aboveskip=2pt,belowskip=2pt,mathescape=true]
class ConcatIterator[A](var iterators: mutable.List[IterableOnce[A]^]):
  // concatenate `it` with `this` into a new iterator
  def concat(it: IterableOnce[A]^): ConcatIterator[A]^{this, it} =$\;$iterators$\;$++=$\;$it.unsafeAssumePure
\end{lstlisting}
Similar to the mutable collectors in Section~\ref{sec:mutable-collectors}, such lists require an explicit capture parameter
to track its elements' capture sets. To preserve compatibility, we instead opt for external tracking:
moving the capture set of the list to the \lstinline|ConcatIterator| itself. This is unsafe (\lstinline{concat}
returns a new iterator reference with expanded captures, but the old reference is still valid).
Future work aims to lift this restriction safely.

Furthermore, interoperability between modules using \lstinline|lib-cc| and \lstinline|lib|
necessitates strict binary compatibility: no new classes or hierarchy changes.
Fortunately, most collection classes are individually capture-checked without such changes.
One notable example is \lstinline{IndexedSeqView[A]} - a lazy view on a sequential collection -
unsoundly implementing non-capturing operations similar to \lstinline{List.map}. To resolve this, we introduced a
separate class hierarchy, achieving capture soundness at the cost of
potential link-time errors, which are easily identified during development.

\section{Discussion}\label{sec:limitations}\label{sec:discussion}

\pparagraph{Towards fully path-dependent capabilities}
Reach capabilities alone are coarse-grained, because they track all the capabilities of a whole data
structure rather than the capabilities of individual elements.
E.g., they do not reflect that a function only accesses the first element of a pair:
\begin{lstlisting}[language=Dotty]
def fst(@use p: (() => Unit, () => Unit)): () ->{p*} Unit = p._1
val p: (() ->{io} Unit, () ->{fs} Unit) = ...
fst(p) // : () ->{io, fs} Unit instead of () ->{io} Unit
\end{lstlisting}
Passing \lstinline|p| to \lstinline|takeFirst| yields a result of type \lstinline|() ->{io, fs} Unit|,
capturing capabilities from both components even though only the first component is returned.
As a workaround, a precise version can be defined in our system by falling back to explicit capture
polymorphism:
\begin{lstlisting}[language=Dotty]
def fstAlt[C1, C2](p: (() ->{C1} Unit, () ->{C2} Unit)): () ->{C1} Unit
\end{lstlisting}
This approach, however, is antithetical to the lightweight nature of our system, which aims to avoid explicit polymorphism in most cases.
Has the lightweight syntax run out of steam already?

Fortunately, the expressive power of reach capabilities and the lightweight syntax can be greatly amplified through Scala's
fully \emph{path-dependent types}, generalizing reach capabilities on one variable to full paths.
We can give precise types to functions like \lstinline|fst| and \lstinline|snd| without resorting
to explicit capture polymorphism by capturing the paths to the individual components:
\begin{lstlisting}[language=Dotty,keepspaces=true,mathescape=true]
def fst( @use p:$\;$(() => Unit,$\;$() => Unit)):$\;$() ->{p._1*} Unit = p._1
def snd( @use p:$\;$(() => Unit,$\;$() => Unit)):$\;$() ->{p._2*} Unit = p._2
def copy(@use p:$\;$(() => Unit,$\;$() => Unit)):$\;$(() ->{p._1*} Unit,$\;$() ->{p._2*} Unit) = (fst(p),snd(p))
\end{lstlisting}
Indeed, the Scala capture checker already supports full paths in this way.

Scala also supports ``capture-set members,'' enabling
capture polymorphism like in DOT~\cite{DBLP:conf/oopsla/RompfA16}:
\begin{lstlisting}[aboveskip=2pt,belowskip=2pt]
trait CaptureSet { type C^ /* declare capture-set member */ }
def capturePoly(c: CaptureSet)(f: File^{c.C}): File^{c.C} = f
\end{lstlisting}
We leave the exploration of the underlying theory for future work.

\vspace{-5pt}
\pparagraph{Telling a Use from a Mention}
Another limitation of our system is that
it cannot precisely distinguish a ``mention'' of a variable
from an actual ``use'' of it \cite{gordon:LIPIcs.ECOOP.2020.10}.
This stems from the \ruleref{var} rule in Figure~\ref{fig:all-typing}, which always adds referenced variables to the use set.
For example:
\begin{lstlisting}[language=scala]
val f: () => Unit = ...
val g = () => f //  : () ->{f} () ->{f} Unit
\end{lstlisting}
Although \lstinline|g| does not invoke \lstinline|f| and only returns it, \lstinline|f| appears in its capture set, making \lstinline|g| impure. However, an equivalent yet pure version can be defined with explicit eta-expansion:
\begin{lstlisting}[language=scala]
val g1 = () => () => f() //  : () -> () ->{f} Unit
\end{lstlisting}
This alternative specifies the evaluation order more precisely, allowing the outer function to be pure while the inner function captures \lstinline|f|.

Our evaluation (\Cref{sec:evaluation}) suggests this limitation is benign in practice. We encountered no blockers while
porting the Scala standard library. Accurately tracking such fine-grained patterns %
would require a substantially richer system and remains interesting
future work.

\vspace{-5pt}
\pparagraph{Fresh Out of the Box: Reasoning about Aliasing and Separation}
Capture checking already shows promise toward fearless concurrency
\cite{DBLP:journals/pacmpl/XuBO24} and a solution to the ``what-color-is-your-function?''
problem \cite{whatcolorisyourfunction} for Scala~3. To be fully effective in these areas, we plan to add reasoning
about aliasing, separation, and capability exclusivity, leading to a system that
can express Rust-style ownership patterns and advanced uses like safe manual memory management.

Reachability types (RT)~\cite{DBLP:journals/pacmpl/BaoWBJHR21,DBLP:journals/pacmpl/WeiBJBR24} (an
independent line of work, unrelated to the Scala~3 capture-checking effort) demonstrate that capture
tracking supports reasoning about aliasing and separation, powering compiler optimizations for
impure functional DSLs in the LMS library for
Scala~2~\cite{DBLP:journals/pacmpl/BracevacWJAJBR23,DBLP:journals/corr/abs-2309-08118}.

RT is a box-free system\footnote{Also note that our reach capabilities have no analogue in RT due to
	the absence of boxes.} that handles generics by threading explicit type-and-capture quantifiers
through every definition~\cite{DBLP:journals/pacmpl/WeiBJBR24}. It therefore faces the same dilemma
as other box-free polymorphic systems of requiring invasive changes (cf. \Cref{sec:related}). Our
work for Scala 3 retrofits the existing ecosystem as much as possible via boxing, perhaps the only
realistic option to introduce CT into Scala without a complete overhaul.
That would alienate existing users.

Furthermore, Wei et al.~\cite{DBLP:journals/pacmpl/WeiBJBR24} point out an issue of CT's boxing mechanism
and the way it is used for escape-checking (see \Cref{sec:motiv:whats-in-the-box})
in conjunction with freshly allocated objects:
\begin{lstlisting}[language=scala,keepspaces=true,mathescape=true]
withFile("file.txt"): f =>
  val r = new Ref(0) // Ref[Int]^, a freshly allocated object, would be Ref[Int]$^{\color{ACMDarkBlue}\vardiamondsuit}$ in RT
  r                  // error, cannot unbox Ref[Int]^ in CT, but legal in RT
\end{lstlisting}
The most natural capture set for fresh values in CT is the top capability \lstinline|cap|,
but then it is impossible for such values to escape the scope they were created in. Hence, \citeauthor{DBLP:journals/pacmpl/WeiBJBR24}
reject having \lstinline|cap| and instead assign a ``freshness marker''
$\vardiamondsuit$ to fresh values, which is not a subcapturing top element.

System \capless{} (\Cref{sec:capless}) replaces the top capability \lstinline|cap| with
explicit quantification under the hood. This yields a principled basis for freshness: whereas
\lstinline|cap| denotes ``arbitrary unknown capabilities,'' we can extend existential quantifiers to distinguish
\emph{arbitrary} from \emph{fresh} capabilities.

For instance, with such extensions,
the function \lstinline|mkRef| that creates a fresh \lstinline|Ref[Int]| can be typed as
\lstinline[mathescape=true]|Int -> $\EXFRESH{c}$ Ref[Int]^{$c$}|,
where {\footnotesize$\EXFRESH{c}T$} represents a fresh existential type.
In the surface language, the type can be simply \lstinline|Int -> Ref[Int]^|,
with the hat \lstinline|^| denoting an existential.
This would enable the following example:
\begin{lstlisting}[language=Dotty,mathescape=true]
def makeCounter(init: Int): $\EXFRESH{c}$ Pair[box () ->{$c$} Int, box () ->{$c$} Int] =
  val r = mkRef(init)
  Pair(box(() => r.get), box(() => r.set(r.get + 1)))
val counter = makeCounter(0)  // c: Fresh, counter: Pair[box () ->{c} Int, box () ->{c} Int]
val get = unbox(counter.snd)  // : () ->{counter} Int  // ok
\end{lstlisting}
The last line is legal because \lstinline|c| is \emph{fresh},
which means that it represents capabilities
that are guaranteed to not conflict with any existing scope,
making unboxing safe without escape concerns.
This example is also supported by RT through their self-references \cite{DBLP:journals/pacmpl/WeiBJBR24},
which allow function types to refer to their own capture sets.
While both approaches express this pattern,
existential quantification arguably provides a more natural and principled representation.

To summarize,
CT embed capturing types
via boxing, whereas RT rely on explicit type-and-capture parameters threaded through
every generic definition, and a different subcapturing relation.  The approaches are therefore already orthogonal, and the gap will widen as we
progress on adding a separation-checking layer on top of System Capless in the future.

\vspace{-5pt}
\pparagraph{Categorizing Capabilities}
Our system covers Scala’s standard \emph{collections} library, but not yet the \emph{full} standard library. 
A notable gap is \smallcode{Try[T]}, a sum of a value \lstinline|T| and a (pure) exception. Its constructor accepts a by-name \lstinline|T|, evaluates it, and captures any thrown value; clients later call \lstinline|.get| to obtain \lstinline|T| or rethrow. 
With the current system we type:
\begin{lstlisting}[language=scala,keepspaces=true]
def apply[T](body: => T): Try[T] // should be Try[T]^{body.only[Control]}
\end{lstlisting}
Here \lstinline|body| is (unboundedly) capture-polymorphic while \lstinline|Try[T]| is pure.
Although that matches the value-level view (a boxed \lstinline|T| or a pure exception), it breaks
the exception discipline enforced by Scala's \lstinline|CanThrow| capabilities
\cite{canthrow,DBLP:conf/scala/OderskyBBLL21}: \lstinline|.get| can rethrow without the
corresponding capability. From the capability view, the constructor should instead record a
\emph{subset} of \lstinline|body|'s captures related to exception-like control transfers: a class we
call \emph{Control Capabilities}. Expressing this requires extending CC with (1) \emph{capability
categories}, and (2) \emph{category-restricted subsets} (e.g.\ \lstinline|body.only[Control]|).

Such categories arise elsewhere: label-safe delimited continuations require restricting captured labels
\cite{10.1145/3679005.3685979}; thread spawning must forbid capturing thread-local resources (e.g.\
mutex guards, stack pointers, non-atomic counters); and mutability can be treated as its own
category. We therefore see capability categorization, especially for asynchronous code, as promising
future work.

\section{Related Work}\label{sec:related}

To our knowledge, this is the first work that introduces an opt-in type system for tracking
effects-as-capabilities into a mainstream programming language, applying it to a widely-used
standard collections library. Although prior works have addressed effect tracking, capabilities, or
resources, they typically focus on effects or ownership rather than directly capturing retained
capabilities. Concerns about retrofitting such systems into existing languages (especially in a
non-invasive and pragmatic manner as in our work) are rarely addressed in the literature.

\vspace{-5pt}
\pparagraph{Polymorphism across Capabilities, Effects, and Ownership}
Polymorphism is indispensable: abstractions must work no matter which effect, resource, or ownership
context they run in. Classic polymorphic effect systems
\cite{DBLP:conf/popl/LucassenG88,DBLP:journals/iandc/TalpinJ94} show how to track side effects,
while region typing \cite{DBLP:journals/iandc/TofteT97,DBLP:conf/pldi/GrossmanMJHWC02}, uniqueness
and linearity
\cite{DBLP:conf/ifip2/Wadler90,DBLP:journals/mscs/BarendsenS96,DBLP:conf/esop/MarshallVO22}, and
ownership types
\cite{DBLP:series/lncs/ClarkeOSW13,DBLP:conf/ecoop/NobleVP98,DBLP:conf/popl/BoyapatiLS03,DBLP:conf/oopsla/PotaninNCB06,DBLP:conf/ecoop/DietlDM07}
control memory and aliasing.  Yet all of these scatter extra indices, regions, uniqueness flags, or
ownership parameters throughout every signature.
Our capture-checking design avoids that clutter: it reuses Scala’s path-dependent types for implicit
polymorphism, and embeds capture-tracked types through boxing, inspired by
CMTT~\cite{DBLP:journals/tocl/NanevskiPP08} and
Effekt~\cite{DBLP:journals/pacmpl/BrachthauserSLB22}. %
Boxing gently embeds a new type universe into the existing one: existing type variables in old
signatures can range over capturing types while preserving parametricity. Consequently, the standard
collections API needs no new type parameters, adoption is strictly opt-in, and code that ignores
capture checking remains fully source-compatible. Optional explicit capture polymorphism is
supported but rarely needed in practice (\Cref{sub:porting-stdlib}).

\vspace{-5pt}
\pparagraph{Practical Ownership and Capability Languages}
Rust~\cite{DBLP:conf/sigada/MatsakisK14} employs an ownership
model to ensure memory safety without
garbage collection. Its type system infers ownership, borrowing, and lifetimes but imposes
restrictive invariants, complicating higher-order functional programming as well as implementations
of data structures, like doubly-linked lists.
Pony~\cite{pony,DBLP:conf/agere/ClebschDBM15,DBLP:journals/pacmpl/ClebschFDYWV17} uses object and
reference capabilities to control mutability and aliasing, achieving capability polymorphism through
complex viewpoint adaptation rules. Mezzo~\cite{DBLP:journals/toplas/BalabonskiPP16} tracks shared,
exclusive, and immutable ownership implicitly via permissions, but still requires explicit
permission polymorphism for generic signatures. Project Verona~\cite{verona} simplifies memory
management with hierarchical regions and ownership semantics. It manages capabilities via isolated
regions and viewpoint adaptation through method overloading.

These systems exemplify qualified type
systems~\cite{DBLP:journals/pacmpl/LeeZLYSB24,DBLP:conf/pldi/FosterFA99,DBLP:journals/scp/Jones94}.
Our work differs by qualifying types explicitly with retained capabilities rather than
permissions or regions, aligning more closely with pure object
capabilities~\cite{objectcapabilites}. The primary novelty compared to previous
CT work~\cite{DBLP:journals/toplas/BoruchGruszeckiOLLB23} is introducing reach capabilities, for which
we are hard-pressed to find a direct analogue in the literature.

\vspace{-6pt}
\pparagraph{OCaml Modal Types}
Efforts to bring Rust-style ownership to
OCaml~\cite{DBLP:journals/corr/abs-2407-11816,oxidizingocaml} adopt a modal type system that tracks
an ambient effect context for enabling stack allocation and preventing data races.
\Cref{sec:discussion} sketches how Scala's capture checking can offer similar guarantees. Like our
work, these approaches aim for backward compatibility and lightweight signatures, yet they rely on
Hindley-Milner inference rather than our local type inference with path-dependent types. Tang et
al.~\cite{DBLP:journals/corr/abs-2407-11816} extend the modal approach to Frank-style effect
handlers~\cite{DBLP:conf/popl/LindleyMM17,DBLP:journals/jfp/ConventLMM20}, tracking effects via the
modal context rather than capture sets. Both systems sometimes require explicit polymorphism.
We look forward to the final
implementation and to clarifying the precise connection between modal and capturing types.

\vspace{-6pt}
\pparagraph{Path-dependent Types and Type members for Effects and Capabilities}
Wyvern
\cite{DBLP:journals/toplas/MelicherXZPA22,DBLP:conf/ecoop/MelicherSPA17,DBLP:conf/oopsla/FishMA20}
builds on object-capabilities \cite{objectcapabilites}, first estimating
authority with a whole-program analysis \cite{DBLP:conf/ecoop/MelicherSPA17} and later adding a
path-dependent effect system
\cite{DBLP:journals/toplas/MelicherXZPA22} inspired by DOT~\cite{DBLP:conf/oopsla/RompfA16}.  Each object can declare an abstract effect member that
upper- or lower-bounds the effects its methods may perform, mirroring DOT’s abstract type
members. Our capture checking takes the complementary view: we record the objects a value can reach
in its capture set; the allowable effects follow implicitly from those objects’ APIs.
Thus both systems use
the same abstraction and composition machinery, but Wyvern names effects directly, while Scala names
the resources through which those effects can occur.

Similarly, associated effects for Flix~\cite{DBLP:journals/pacmpl/LutzeM24}
extend the same idea to type classes: each class can declare an abstract effect row, giving first-class, per-instance effect polymorphism without extra type
parameters.  Like Wyvern, it abstracts over effects; capture sets instead abstract over
reachable resources, thereby opening the door to providing alias-control guarantees (cf. \Cref{sec:discussion}).

\vspace{-6pt}
\pparagraph{Second-Class Values}
Osvald et al.~\cite{osvald2016gentrification} introduced 2nd-class values for Scala~2, i.e., values
whose lifetime is tied to their lexical scope, but their stratification hinders higher-order
programming, restricting currying and lazy collections. Xhebraj et
al.~\cite{DBLP:conf/ecoop/XhebrajB0R22} extend the idea with 2nd-class returns (recovering currying)
and explicit storage-mode polymorphism, yet rely on an unusual stack semantics.
Being based on Scala~2 annotations, their implementation cannot
express fully storage-mode generic types such as \lstinline|List[Q, T @mode[Q]]|.
Our boxing-based system avoids these limitations. It does not yet support privilege lattices
or bounding a closure’s free variables, features useful for fractional capabilities
\cite[Fig.~7]{DBLP:conf/ecoop/XhebrajB0R22}, which we expect to model via capability categories
(\Cref{sec:discussion}).

\section{Conclusion}\label{sec:conclusion}

This paper introduced reach capabilities and developed their foundations, backed by mechanized type soundness and scope safety
proofs. Reach capabilities
made it possible to migrate the Scala 3 standard collections library to capture tracking with minimal adjustments.
Our approach effectively addresses limitations in
handling capabilities within generic data structures and collections, a significant milestone towards bringing lightweight and ergonomic effect systems
to the masses.
 
\clearpage

\begin{acks}
We thank Eug\`{e}ne Flesselle and the anonymous reviewers for their valuable feedback.
This work is supported by the \grantsponsor{snsf}{SNSF}{https://www.snf.ch} Advanced Grant \grantnum{snsf}{209506}, ``Capabilities for Typing Resources and Effects''.
\end{acks}

\section*{Data-Availability Statement}

The artifacts accompanying this paper \cite{artifact} include: (1) a complete mechanized formalization of System
Capless in Lean 4, including the scope safety extension and associated metatheory; (2) a practical
implementation of capture checking integrated into the Scala 3 compiler; and (3) source code and
scripts for reproducing the paper’s case studies.
The capture-checked version of
Scala’s standard collections library is included among these artifacts.

 \printbibliography{}

\clearpage
\appendix
\section{Additional Definitions}

\subsection{Well-formedness}
\label{sec:well-formedness}

\begin{figure}[htbp]
\centering  
\footnotesize
\begin{multicols}{3}

\infax[\ruledef{wf-top}]
{\wf{\G}{\top}}

\infrule[\ruledef{wf-tvar}]
{X<:S\in\G}
{\wf{\G}{X}}

\infrule[\ruledef{wf-fun}]
{\wf{\G}{T}\andalso\wf{(\G,x:T)}{E}}
{\wf{\G}{\forall(x:T)E}}

\infrule[\ruledef{wf-tfun}]
{\wf{\G}{S}\andalso\wf{(\G,X<:S)}{E}}
{\wf{\G}{\forall[X<:S]E}}

\infrule[\ruledef{wf-cfun}]
{\wf{\G}{B}\andalso\wf{(\G,c<:B)}{E}}
{\wf{\G}{\forall[c<:B]E}}

\infrule[\ruledef{wf-exists}]
{\wf{(\G,c<:*)}{T}}
{\wf{\G}{\EXCAP{c}T}}

\infrule[\ruledef{wf-capt}]
{\wf{\G}{S}\andalso\wf{\G}{C}}
{\wf{\G}{S\capt C}}

\infrule[\ruledef{wf-cset}]
{C\subseteq\dom{\G}}
{\wf{\G}{C}}

\infax[\ruledef{wf-star}]
{\wf{\G}{*}}
\end{multicols}

\caption{Definition of well-formedness of System \capless{}.}
\label{fig:well-formedness}
\end{figure}
 
\Cref{fig:well-formedness} defines the well-formedness of types, capture sets and capture bounds in System \capless{}.
It basically ensures that the capture sets only contain defined variables in the context.

\begin{figure}[htbp]
\centering  
\footnotesize
\begin{multicols}{3}

\infax[\rruledef{wf-top}]
{\wf{\G}{\top}}

\infrule[\rruledef{wf-tvar}]
{X<:\top\in\G}
{\wf{\G}{X}}

\infrule[\rruledef{wf-fun}]
{\wf{\G}{T}\andalso\wf{(\G,x:T)}{U}}
{\wf{\G}{\forall^\alpha(x:T)U}}

\infrule[\rruledef{wf-tfun}]
{\wf{(\G,X<:\top)}{T}}
{\wf{\G}{\forall[X<:\top]T}}

\infrule[\rruledef{wf-cfun}]
{\wf{(\G,c)}{T}}
{\wf{\G}{\forall[c]T}}

\infrule[\rruledef{wf-boxed}]
{\wf{\G}{T}}
{\wf{\G}{\BOX T}}

\infrule[\rruledef{wf-applied}]
{\kappa=(\kappa_1,\cdots,\kappa_n)\mapsto S\in\cctx\\
 \forall T_i, \wf{\G}{T_i}}
{\wf{\G}{\kappa[T_1,\cdots,T_n]}}

\infrule[\rruledef{wf-capt}]
{\wf{\G}{S}\andalso\wf{\G}{C}}
{\wf{\G}{S\capt C}}

\infrule[\rruledef{wf-cset}]
{\forall \theta\in C. \wf{\G}{\theta}}
{\wf{\G}{C}}

\infax[\rruledef{wf-cap}]
{\wf{\G}{\CAP}}

\infrule[\rruledef{wf-var}]
{x:T\in\G}
{\wf{\G}{x}}

\infrule[\rruledef{wf-reach}]
{x:T\in\G}
{\wf{\G}{x^*}}

\end{multicols}

\caption{Definition of well-formedness of System \cappy{}.}
\label{fig:cappy-well-formedness}
\end{figure}
 
\Cref{fig:cappy-well-formedness} shows the definition of the well-formedness judgement in System \cappy{}.
It in most parts matches that of System \capless{}
and is thus not surprising.
The primary difference is that
capture sets are allowed to include reach capabilities of defined variables
and the universal capability $\CAP$.

\begin{figure*}[htbp]
\centering  
\footnotesize

\flushleft{\textbf{Well-Formedness of Type Definition \quad $\wfv{\Theta;\G;(X_1^{\nu_1},\cdots,X_n^{\nu_n})}{S}{\nu}$}}

\begin{multicols}{2}

\infax[\rruledef{d-top}]
{\wfv{\Theta;\G;(X_1^{\nu_1},\cdots,X_n^{\nu_n})}{\top}{\nu}}

\infrule[\rruledef{d-capt}]
{\wf{\G}{C}\andalso\CAP\notin C\text{ if $\nu = +$}\\
 \wfv{\Theta;\G;(X_1^{\nu_1},\cdots,X_n^{\nu_n})}{S}{\nu}
}
{\wfv{\Theta;\G;(X_1^{\nu_1},\cdots,X_n^{\nu_n})}{S\capt C}{\nu}}

\infrule[\rruledef{d-tvar}]
{\wf{\G}{X}}
{\wfv{\Theta;\G;(X_1^{\nu_1},\cdots,X_n^{\nu_n})}{X}{\nu}}

\infax[\rruledef{d-tparam}]
{\wfv{\Theta;\G;(\cdots,X_i^{\nu_i},\cdots)}{X_i}{\nu_i}}

\infrule[\rruledef{d-fun}]
{\wfv{\Theta;\G;(X_1^{\nu_1},\cdots,X_n^{\nu_n})}{T}{\neg\nu}\\
 \wfv{\Theta;\G,x:T;(X_1^{\nu_1},\cdots,X_n^{\nu_n})}{U}{\nu}}
{\wfv{\Theta;\G;(X_1^{\nu_1},\cdots,X_n^{\nu_n})}{\forall^\alpha(x:T)U}{\nu}}

\infrule[\rruledef{d-tfun}]
{\wfv{\Theta;\G,X<:\top;(X_1^{\nu_1},\cdots,X_n^{\nu_n})}{T}{\nu}}
{\wfv{\Theta;\G;(X_1^{\nu_1},\cdots,X_n^{\nu_n})}{\forall[X<:\top]T}{\nu}}

\infrule[\rruledef{d-cfun}]
{\wfv{\Theta;\G,c;(X_1^{\nu_1},\cdots,X_n^{\nu_n})}{T}{\nu}}
{\wfv{\Theta;\G;(X_1^{\nu_1},\cdots,X_n^{\nu_n})}{\forall[c]T}{\nu}}

\infrule[\rruledef{d-boxed}]
{\wfv{\Theta;\G;(X_1^{\nu_1},\cdots,X_n^{\nu_n})}{T}{\nu}}
{\wfv{\Theta;\G;(X_1^{\nu_1},\cdots,X_n^{\nu_n})}{\BOX T}{\nu}}

\infrule[\rruledef{d-applied}]
{\kappa=(Y_1^{\nu'_1},\cdots,Y_n^{\nu'_m})\mapsto S\in\cctx\\
 \forall Y_i^+,\wfv{\Theta;\G;(X_1^{\nu_1},\cdots,X_n^{\nu_n})}{T_i}{\nu}\\
 \forall Y_i^-,\wfv{\Theta;\G;(X_1^{\nu_1},\cdots,X_n^{\nu_n})}{T_i}{\neg\nu}}
{\wfv{\Theta;\G;(X_1^{\nu_1},\cdots,X_n^{\nu_n})}{\kappa[T_1,\cdots,T_m]}{\nu}}

\end{multicols}

\flushleft{\textbf{Well-Formedness of Typedef Context \quad $\wf{}{\cctx}$}}

\begin{multicols}{2}

\infax[\rruledef{ds-empty}]
{\wf{}{\emptyset}}

\infrule[\rruledef{ds-cons}]
{\wf{}{\Theta}\\
 \wfv{\Theta;\emptyset;(X_1^{\nu_1},\cdots,X_n^{\nu_n})}{S}{+}}
{\wf{}{\Theta,\kappa=(X_1^{\nu_1},\cdots,X_n^{\nu_n})\mapsto S}}

\end{multicols}

\caption{Well-formedness of type definitions.}
\label{fig:cappy-well-formedness-tycon}
\end{figure*} \Cref{fig:cappy-well-formedness-tycon} defines the well-formedness judgement
for type definitions.
$\wf{}{\cctx}$ ensures that all type definitions in the context are well-formed.
Later type definitions can refer to earlier ones.
$\wfv{\cctx;\G;(X_1^{\nu_1},\cdots,X_n^{\nu_n})}{T}{\nu}$ ensures that
the body of a type definition is well-formed under
the variance $\nu$,
the existing type definitions $\cctx$,
a local type context $\G$,
and the parameter clause $(X_1^{\nu_1},\cdots,X_n^{\nu_n})$.
It extends type well-formedness by additionally ensuring that
the variances of parameters are consistent (as in \rruleref{d-tvar})
and that the body does not have covariant occurrences of $\CAP$s.

\subsection{Reduction Rules of System \capless{}}
\label{sec:reduction}

\begin{figure*}[htbp]
\scriptsize

\flushleft{\textbf{Syntax}}

\begin{align*}
			\sta\coloneqq\ &           \tag*{\textbf{Store}}\\
			&\emptyset                    \tag*{empty}\\
			&\sta, \VAL x\mapsto v                    \tag*{val binding}\\
			\ec\coloneqq\ &           \tag*{\textbf{Evaluation Context}}\\
			&[]                    \tag*{hole}\\
			&\LET x = \ec \IN t           \tag*{let}\\
			&\new{\LET \< c,x \> = \ec \IN t}           \tag*{ex. let}\\
\end{align*}

\flushleft{\textbf{Reduction \quad $\<\sta\BAR t\>\red\<\sta'\BAR t'\>$}}

\begin{multicols}{2}

\infrule[\ruledef{apply}]
{\sta(x) = \lambda(z: T)t}
{\<\sta\BAR\ec[ x\,y ]\> \red \<\sta\BAR\ec[ [z:=y]t ]\>}

\infrule[\ruledef{tapply}]
{\sta(x) = \lambda[X<:S]t}
{\<\sta\BAR\ec[ x[S'] ]\> \red \<\sta\BAR\ec[ [X:=S']t ]\>}

\newruletrue
\infrule[\ruledef{capply}]
{\sta(x) = \lambda[c<:B]t}
{\<\sta\BAR\ec[ x[C] ]\> \red \<\sta\BAR\ec[ [c:=C]t ]\>}
\newrulefalse

\infax[\ruledef{rename}]
{\<\sta\BAR\ec[ \LET x = y \IN t ]\>\red 
 \<\sta\BAR\ec[ [x:=y]t ]\>}

\newruletrue
\infax[\ruledef{rename-e}]
{\<\sta\BAR\ec[\LET \langle c, x \rangle =\< C,y \> \IN t]\>\red 
 \<\sta\BAR\ec[ [x:=y][c:=C]t]\>}
\newrulefalse

\infax[\ruledef{lift}]
{\<\sta\BAR\ec[ \LET x = v\IN t ]\> \red 
 \<\sta, \VAL x\mapsto v\BAR\ec[ t ]\>}

\end{multicols}

\vspace{-1.5em}
\caption{Evaluation rules of System \capless{}. Changes from System \ccformal{} \cite{DBLP:journals/toplas/BoruchGruszeckiOLLB23} are highlighted.}\label{fig:reduction}

\end{figure*}
 
Figure \ref{fig:reduction} presents the reduction rules of System \capless{}.
It is the same to those of System \ccformal{}
except for the addition of rules for capture application \ruleref{capply} and existential unpacking \ruleref{rename-e}.

\subsection{Dynamic Rules of Scoped Capabilities}
\label{sec:dynamic-rules-scopedcaps}

\begin{figure*}[tbp]
\footnotesize

\begin{minipage}{.48\textwidth}
\flushleft{\textbf{Syntax}}

\begin{align*}
  l_S         \tag*{\textbf{Label}}\\
  p,q,r\coloneqq\ & l_S \mid x           \tag*{\textbf{General Variable}}\\
  \theta\coloneqq\ & \new{p} \mid c           \tag*{\textbf{Capture}}\\
  s,\,t,\,u\coloneqq\ & \SCOPE_{l_S} \IN t \mid \cdots           \tag*{\textbf{Term}}\\
  \ec\coloneqq\ & \SCOPE_{l_S}\IN\ec \mid \cdots         \tag*{\textbf{Continuation Stack}}\\
\end{align*}

\end{minipage}\quad%
\begin{minipage}{.48\textwidth}

\flushleft{\textbf{Typing} \quad $\typs{\G}{t}{T}$}

\begin{multicols}{2}

\infax[\ruledef{label}]
{\typ{\set{l_S}}{\G}{l_S}{\BREAK[S]\capt\set{l_S}}}

\infrule[\ruledef{scope}]
{\typ{C}{\G}{t}{S}}
{\typ{C}{\G}{\SCOPE_{l_S}\IN t}{S}}

\end{multicols}

The \ruleref{var}, \ruleref{app}, \ruleref{tapp}, \ruleref{tapp}, \ruleref{invoke}, and \ruleref{pack} rules now work on general variables $p$.

\flushleft{\textbf{Subcapturing} \ $\subs{\G}{C_1}{C_2}$} 
now works on general variables $p$

\end{minipage}

\vspace{0.5em}

\flushleft{\textbf{Reduction \quad $\<\sta\BAR t\>\red\<\sta'\BAR t'\>$}}

\infrule[\ruledef{enter}]
{\text{$l$ is fresh}}
{\<\sta\BAR\ec[ \BOUNDARY[S]\AS\<c,x\>\IN t ]\>\red
 \<\sta\BAR\ec[ \SCOPE_{l_S}\IN [c:=\set{l_S}][x:=l_S]t ]\>}

\begin{multicols}{2}
\infrule[\ruledef{breakout}]
{\ec = \ec_1[ \SCOPE_{l_S}\IN \ec_2 ]}
{\<\sta\BAR\ec[ l_S\,p ]\>\red
 \<\sta\BAR\ec_1[ p ]\>}

\infax[\ruledef{leave}]
{\<\sta\BAR\ec[ \SCOPE_{l_S}\IN a ]\>\red
 \<\sta\BAR\ec[ a ]\>}
\end{multicols}

The \ruleref{apply}, \ruleref{tapply}, \ruleref{capply}, \ruleref{rename}, and \ruleref{rename-e} rules now work on general variables $p$.

\caption{Extensions to dynamic rules of \calculus{}.}\label{fig:dynamic-extensions}
\end{figure*} 
Figure~\ref{fig:dynamic-extensions} shows the extensions to the dynamic rules of System \capless{}.
We introduce runtime labels $l_S$ and scopes $\SCOPE_{l_S}\IN t$
for the dynamic semantics of scoped capabilities.
A label $l_S$ is a unique runtime identifier of a boundary.
\ruleref{enter} reduces a boundary form
by generating a fresh label $l$ as the identifier of the boundary.
The label also becomes the representation of the $\BREAK$ capability.
The boundary form is then replaced by a scope form $\SCOPE_{l_S}\IN t$,
which delimits the scope of the boundary.
The scope can be left by returning an answer $a$ from it,
as shown in the \ruleref{leave} rule.
It can also be broken out by invoking the $\BREAK$ capability,
as shown in the \ruleref{break} rule,
which looks for the scope with the matching label $l_S$ in the evaluation context,
and transfers the control there.
If a $\BREAK$ capability is invoked outside of its scope,
the evaluation gets stuck.
\section{Discussions}

\subsection{An Example of Local Mutable State Storing Tracked Values}
\label{sec:local-mutable-state}

The following defines the \lstinline|concat| function 
which takes a list of iterators
and merges them into one single iterator:
\begin{lstlisting}[language=Dotty]
def concat[T](@use xs: List[Iterator[T]^]): Iterator[T]^{xs*} = new Iterator[T]:
  var rest: List[Iterator[T]^{xs*}] = xs
  var cur: Iterator[T]^{xs*} = Iterator.empty
  def hasNext(): Boolean = cur.hasNext() || rest.exists(_.hasNext())
  def next(): T =
    if cur.hasNext() then
      cur.next()
    else
      cur = rest.head
      rest = rest.tail
      next()
\end{lstlisting}
It has the local mutable states \lstinline|rest| and \lstinline|cur|,
both store impure values.
In CT, the type of a mutable state cannot contain any covariant mention of $\CAP$,
as that provides a side channel for scoped capabilities to escape \cite{DBLP:journals/toplas/BoruchGruszeckiOLLB23}.
Previously, the above programming pattern is not supported,
as the best type we can assign to the element of \lstinline|xs| is \lstinline|Iterator[T]^{cap}|,
and storing a value of this type violates the aforementioned restriction.
Reach capabilities enable this pattern,
since we have a precise name \lstinline|xs*| other than the top element $\CAP$ for the effects of the elements in the list.

\subsection{Theoretical Benefits of System \capless{}}
\label{sec:capless:theoretical-benefits}

\subsubsection{Farewell, Boxes}\label{sec:capless-box}

System \capless{} drops boxes.
In fact, 
thanks to precise capture tracking over curried functions,
boxes can be encoded in System \capless{} using term abstractions.
Specifically,
a boxed term $\BOX x: \BOX S\capt C$ can be encoded as
a double type lambda abstraction
$\lambda[X<:\top]\lambda[X<:\top]. x$.
The inner lambda wraps the variable $x$ as a value with an empty use set:
it can be typed as $\typ{\set{}}{\G}{\lambda[X<:\top].x}{(\forall[X<:\top]S\capt C)\capt C}$.
Note that this value still has a non-empty capture set.
The outer lambda wraps it again to obtain a pure value:
$\typ{\set{}}{\G}{\lambda[X<:\top].\lambda[X<:\top].x}{(\forall[X<:\top](\forall[X<:\top]S\capt C)\capt C)\capt\set{}}$.
Then, the unbox term $C'\UNBOX x$
can be encoded as $x[\top][\top]$.
For the simplicity of presentation,
this term is not in MNF but can be normalized to that form trivially.
This way, boxing and unboxing behaves the same as in System \cappy{}
boxing hides the captures of a variable
and unboxing reveals them.
This is exactly how the translation system translates boxes in System \cappy{} to System \capless{}.
For more details, see Appendix~\ref{sec:decap}.

\subsubsection{Justification of Escape Checking}
\label{sec:justifying-unboxes}

System \capless{} provides insights into the unboxing restriction for escape checking proposed in the original System \ccformal{} \cite{DBLP:journals/toplas/BoruchGruszeckiOLLB23},
which has the restriction that
given a boxed value of type \lstinline|box T^{x1,...xn}|
it can be unboxed only if there is no \lstinline|cap|
in \lstinline|{x1,...,xn}|.
This is for \emph{escape checking}:
a scoped capability should not be used outside its defining scope.
Although this restriction is sound and useful, it feels somewhat
ad-hoc.
System \capless{} provides insights
and a theoretical justification
of this restriction.
In short,
\TCAP{} means existentially quantified capture sets,
and unboxing a \TCAP{} indicates using capabilities
that are unknown in the current scope,
which should be rejected.

The following is a concrete example
of an attempt to use a scoped capability outside its scope:
\begin{lstlisting}[language=Dotty]
def withFile[T](path: String)(op: (f: File^) => T): T = ...
val leaked = withFile[box () ->{cap} String]("test.txt"): f =>
  box () => f.read()
(unbox leaked)()  // error: cannot unbox with `cap'
\end{lstlisting}
\begin{itemize}
\item \lstinline|withFile| creates a scoped file capability given a path
and allows the usage of the file inside the scope of the function \lstinline|op|.
Afterwards the file is closed.
\item In the \lstinline|leaked| definition,
the user tries to return a closure that reads the file,
letting the scoped capability escape.
\item This program is rejected by capture checking,
as \lstinline|leaked| has the type \lstinline|box () ->{cap} String|,
and therefore is forbidden to be unboxed.
\end{itemize}
Intuitively,
the $\CAP$ in the type means that there are certain capabilities widen to the top,
which happens typically when a capability escapes its defining scope.

The equivalent of the previous example in System \capless{} is:
\begin{align*}
&\LET \<c,\textsf{leaked}\> = \textsf{withFile}[\BOX (() \rightarrow \textsf{String})\capt\set{c}](\textsf{"test.txt"})(f \Rightarrow \BOX ()\Rightarrow f.\textsf{read}()) \IN\\
&\LET z_0 = \set{c}\UNBOX \textsf{leaked} \IN z_0()
\end{align*}
Here, $\LET\<c,x\>=\cdots\IN\cdots$
is the form for unpacking an existential quantification.
Given a value of the type $\EXCAP{c'}T$,
this form binds the existential witness as $c$
and binds the value as $x$ where $c'$ is replaced by $c$ in $T$.
This form has the restriction that
the existential witness $c$ cannot be appear in the captured variables of the continuation term,
since this behavior signifies the capabilities witnessed by $c$ is used in the body
and we need a way to account for the capturing of these capabilities.
But we have no ways to approximate widen $c$
since it is completely abstract.
And the previous example violates this restriction:
the unbox operation $\set{c}\UNBOX\textsf{leaked}$
adds $c$ to the captured variables of the continuation term.
This signfies that some unknown and probably out-of-scope capabilities are used,
justifying the rejection of the program.

\subsection{Reach Refinement and Function Types}
\label{sec:reach-refinement-function-types}

Surprisingly, reach refinement \rruleref{r-fun} touches neither the domain nor the codomain of function types.
Let us first clarify why the codomain is not refined.
Assume $\textsf{IO}$ is a type in the context
and consider the example $x: (\forall(z: \textsf{IO}\capt\set{\CAP}) \textsf{IO}\capt\set{\CAP})\capt C$.
If we were to refine the codomain
and type $x$ as $(\forall(z: \textsf{IO}\capt\set{\CAP}) \textsf{IO}\capt\set{x^*})\capt C$,
we would implicitly treat the $\CAP$ in the codomain
as existentially quantified at the same scope of $x$,
namely $\EXCAP{c_{x^*}}\forall(z:\cdots)\textsf{IO}\capt\set{c_{x^*}}$.
This interpretation is incorrect
because that $\CAP$ may depend on capabilities introduced by the argument $z$.
A concrete example is the identity function $\lambda(z: \textsf{IO}\capt\set{\CAP}) z$.
A proper interpretation of $x$'s type is
$\forall(z:\cdots)\EXCAP{c}\textsf{IO}\capt\set{c}$.
Hence, the $\CAP$ in the codomain must not be replaced by $x^*$.

A similar reason applies to why the domain is unchanged.
One may naturally consider refining covariant $\CAP$ occurring in a ``double-flip'' scenario.
Consider an example\footnote{Slightly abusing the notation, we write $T\Rightarrow U$ as a shorthand for $(\forall(x:T)U)\capt\set{\CAP}$ to improve readability.}
$
x: (z_1: \textsf{IO}\capt\set{\CAP}\Rightarrow\textsf{Unit}) \Rightarrow (z_2: \textsf{IO}\capt\set{\CAP}) \Rightarrow \textsf{Unit}.
$
It is tempting to refine the covariantly occurring $\CAP$ in the domain,
yielding the type
$
(z_1: \textsf{IO}\capt\set{x^*}\Rightarrow\textsf{Unit}) \Rightarrow (z_2: \textsf{IO}\capt\set{\CAP}) \Rightarrow \textsf{Unit}.
$
However, this refinement is incorrect because that $\CAP$ might correspond to capabilities introduced by a later argument, such as $z_2$.
The function $\lambda(z_1:\cdots)\lambda(z_2:\cdots)z_1\,z_2$ is such an example.

\section{Proof of System \capless{}}

\subsection{Proof Devices}
\label{sec:capless-proof-devices}

\begin{figure*}[htbp]
\scriptsize

\flushleft{\textbf{Store Typing \quad $\typSta{\sta}{\Delta}$}}

\vspace{0.3em}

\begin{multicols}{2}

\infax[\ruledef{s-empty}]
{\typSta{\emptyset}{\emptyset}}

\infrule[\ruledef{s-val}]
{\typSta{\sta}{\G}\andalso
 \typs{\G}{v}{T}}
{\typSta{\sta, \VAL x\mapsto v}{\G, x: T}}
  
\end{multicols}

\caption{Store Typing}
\label{fig:store-typing}
\end{figure*}
 
Figure \ref{fig:store-typing} presents the store typing rules for System \capless{}.

\section{Proof of System \cappy{} and System \decap{}}
\label{sec:proof}

\subsection{System \decap{}}
\label{sec:decap}

\begin{figure*}[htbp]
\footnotesize

\flushleft{\textbf{Capture Set Encoding \quad $\encode{\<D,\rho,\rho^*\>}{C}{C'}$}}

\begin{multicols}{2}

\infax[\ruledef{i-empty}]
{\encode{\tau}{\set{}}{\set{}}}

\infrule[\ruledef{i-union}]
{\encode{\tau}{C_1}{C'_1}\andalso
 \encode{\tau}{C_2}{C'_2}}
{\encode{\tau}{C_1\cup C_2}{C'_1\cup C'_2}}

\infax[\ruledef{i-var}]
{\encode{\<D,\rho,\rho^*\>}{\set{x}}{\rho(x)}}

\infax[\ruledef{i-cap}]
{\encode{\<D,\rho,\rho^*\>}{\set{\CAP}}{D}}

\infax[\ruledef{i-reach}]
{\encode{\<D,\rho,\rho^*\>}{\set{x*}}{\rho^*(x)}}

\end{multicols}

\flushleft{\textbf{Type Encoding \quad $\encode{\<D,\rho,\rho^*\>}{T}{T'}$}}

\begin{multicols}{2}
  
\infrule[\ruledef{i-capt}]
{\encode{\tau}{S}{S'}\andalso\encode{\tau}{C}{C'}}
{\encode{\tau}{S\capt C}{S'\capt C'}}

\infax[\ruledef{i-top}]
{\encode{\tau}{\top}{\top}}

\infax[\ruledef{i-tvar}]
{\encode{\tau}{X}{X}}

\infrule[\ruledef{i-tfun}]
{\encode{\<D,\rho,\rho^*\>}{T}{T'}}
{\encode{\<D,\rho,\rho^*\>}{\forall[X<:\top]T}{\forall[X<:\top]T'}}

\infrule[\ruledef{i-cfun}]
{\encode{\<D,\rho,\rho^*\>}{T}{T'}}
{\encode{\<D,\rho,\rho^*\>}{\forall[c]T}{\forall[c]T'}}

\infrule[\ruledef{i-boxed}]
{\\ \\ \encode{\tau}{T}{S'\capt C'}}
{\encode{\tau}{\BOX T}{\forall[X<:\top](\forall[X<:\top]S'\capt C')\capt C'}}

\infrule[\ruledef{i-applied}]
{\kappa=(X_1^{\nu_1},\cdots,X_n^{\nu_n})\mapsto S\in\cctx\\
 \forall X_i^+,\encode{\tau}{T_i}{T'_i}\andalso
 \forall X_i^-,T'_i = T_i\\
 \encode{\tau}{[X_1:=T'_1,\cdots,X_n:=T'_n]S}{U}}
{\encode{\tau}{\kappa[T_1,\cdots,T_n]}{U}}

\end{multicols}

\infrule[\ruledef{i-fun}]
{\encode{\<\set{c_{x^*}},\rho,\rho^*\>}{S}{S'}\andalso
 \encode{\<\set{*}, \rho,\rho^*\>}{C_a}{B_a}\\
 \encode{\<\set{c},\rho[x\mapsto \set{c_x}],\rho^*[x\mapsto \set{c_{x*}}]\>}{U}{U'}\andalso
 \encode{\<D,\rho,\rho^*\>}{C_f}{C'_f}\\
 C''_f = \begin{cases}
   C'_f\cup\set{c_x} & \text{if $\alpha = \epsilon$}, \\
   C'_f\cup\set{c_x,c_{x*}} & \text{if $\alpha = \USE$}. \\
 \end{cases}}
{\encode{\<D,\rho,\rho^*\>}{(\forall^\alpha(x:S\capt C_a)U)\capt C_f}{\forall[c_x<:B_a]\forall[c_{x*}](\forall(x:S'\capt\set{c_x})\EXCAP{c} U')\capt C''_f}}

\vspace{-1.5em}
\caption{Semantics of System \decap{}}
\label{fig:cappy-encoding}
\end{figure*}
 
Figure \ref{fig:cappy-encoding} presents System \decap{},
a translation system
that translates \cappy{} types to \calculus{} types.

\begin{definition}[Translation Context]
A translation context, denoted as $\<C,\rho,\rho^*\>$
and by the meta-variable $\tau$,
consists of the following components:
\begin{itemize}
\item a capture set $D$ in System \capless{};
\item a function $\rho$ which maps term variable names to capture sets in System \capless{};
\item and a function $\rho^*$ which maps term and type variable names to capture sets in System \capless{}.
\end{itemize}
\end{definition}

Here, $D$ is called the \emph{interpretation}.
It assigns meanings to the $\CAP$s in the source language.
The $\rho$ function maps capabilities into their underlying capture sets in the target system.
And similarly for $\rho^*$, which maps reach capabilities.

\begin{definition}[Proper Translation Context]
\label{def:proper-tcontext}
A translation context $\tau = \<C,\rho,\rho^*\>$ is \emph{proper}
under a source context $\G$ and a target context $\Delta$, iff
\begin{itemize}
\item $D$, the codomain of $\rho$ and $\rho^*$, are well-formed in $\Delta$;
\item $\dom{\Delta}\subseteq\dom{\rho}$ and $\dom{\Delta}\subseteq\dom{\rho^*}$;
\item $\rho^*$ is a bijection;
\item for any $x:S\capt C\in\G$,
we have $\encode{\<\rho(x),\rho,\rho^*\>}{C}{C'}$
and $\subs{\Delta}{\rho(x)}{C'}$.
\item for any $x:S\capt C\in\G$,
we have $x:S'\capt\rho(x)\in\G$
where $\encode{\<\rho^*(x),\rho,\rho^*\>}{S}{S'}$.
\item for any $X<:S\in\G$,
we have $\rho^*(X) = \set{c_{X}}$
and $X<:S'\in\Delta$
where $\encode{\<\set{c_{X}},\rho,\rho^*\>}{S}{S'}$.
\end{itemize}
\end{definition}

\subsection{Properties of System \decap{}}

\begin{theorem}[Capture Set Translation is Complete]
\label{theorem:capset-interp-complete}
Given any $\<D,\rho,\rho^*\>$ and $C$,
there exists $I$ such that
$\encode{\<D,\rho,\rho^*\>}{C}{I}$.
\end{theorem}

\begin{proof}
By induction on $|C|$, the size of $C$.

\emph{Case $|C| = 0$}.
Then $C$ is empty. Conclude this by the \ruleref{i-empty} rule.

\emph{Case $|C| = 1$}.
Then $C$ is a singleton.
If $C = \set{\CAP}$, set $I = D$ and conclude this case by the \ruleref{i-cap} rule.
If $C = \set{x}$, set $I = \rho(x)$ and this case is concluded by the \ruleref{i-var} rule.
Otherwise, we have $C = \set{x^*}$ for some $x$.
This can be concluded by the \ruleref{i-reach} rule.

\emph{Case $|C| > 1$}.
Then we can split $C = C_1\cup C_2$
where $|C_1| < |C|$ and $|C_2| <: |C|$.
We conclude by the IH and the \ruleref{i-union} rule.
\end{proof}

\begin{theorem}[Capture Set Translation is Monotonic (I)]
\label{theorem:capset-interp-monotonic}
Given $\<D_1,\rho,\rho^*\>$ and $\<D_2,\rho,\rho^*\>$,
two proper translation contexts under $\G$ and $\Delta$,
if
\begin{enumerate}
\item $\subs{\Delta}{D_1}{D_2}$,
\item $\encode{\<D_1,\rho,\rho^*\>}{C}{I_1}$,
\item and $\encode{\<D_2,\rho,\rho^*\>}{C}{I_2}$,
\end{enumerate}
then
$\subs{\Delta}{I_1}{I_2}$
\end{theorem}

\begin{proof}
By induction on the first translation derivation.

\emph{Case \ruleref{i-empty}}.
Then $C = I_1 = \set{}$.
This case can be concluded immediately,
since an empty capture set is a subcapture of any capture set.

\emph{Case \ruleref{i-union}}.
Then $C = C_{1}\cup C_{2}$,
$I_1 = I_{11}\cup I_{12}$,
$\encode{\<D_1,\rho,\rho^*\>}{C_{1}}{I_{11}}$,
and $\encode{\<D_1,\rho,\rho^*\>}{C_{2}}{I_{12}}$.
We can analyze the second translation derivation
and conclude this case by applying the IH and the \ruleref{i-union} rule.

\emph{Case \ruleref{i-var}}.
Then $C = \set{x}$.
We can show that $I_1 = I_2 = \rho(x)$.
We can conclude this case by the reflexivity of the subcapturing relation.

\emph{Case \ruleref{i-reach}}.
Analogous to the \ruleref{i-var} case.

\emph{Case \ruleref{i-cap}}.
Then $C = \set{\CAP}$.
By analyzing the translation derivations,
we can show that $I_1 = D_1$ and $I_2 = D_2$.
So this case follows directly from the assumption.
\end{proof}

\begin{theorem}[Redundant Interpretation]
\label{theorem:interp-redundant}
Given any $\encode{\<D,\rho,\rho^*\>}{C}{I}$
such that $\CAP\notin C$,
we can show that
$\encode{\<\set{},\rho,\rho^*\>}{C}{I}$.
\end{theorem}

\begin{proof}
By induction on the translation derivation.

\emph{Case \ruleref{i-empty}, \ruleref{i-var} and \ruleref{i-reach}}.
These cases follow directly from the premise.

\emph{Case \ruleref{i-union}}.
By the IHs and the same rule.

\emph{Case \ruleref{i-cap}}.
This case is absurd.
\end{proof}

\begin{theorem}[Capture Set Translation is Monotonic (II)]
\label{theorem:capset-interp-monotonic-helper}
Given a proper translation context $\<\set{},\rho,\rho^*\>$
under contexts $\G$ and $\Delta$,
two capture sets $\subs{\G}{C_1}{C_2}$
such that $\CAP\notin C_1$ and $C_2$,
if $\encode{\<\set{},\rho,\rho^*\>}{C_1}{I_1}$
and $\encode{\<\set{},\rho,\rho^*\>}{C_2}{I_2}$,
then $\subs{\Delta}{I_1}{I_2}$.
\end{theorem}

\begin{proof}
By induction on the subcapturing derivation.

\emph{Case \rruleref{sc-trans}}.
Then $\subs{\G}{C_1}{C_0}$
and $\subs{\G}{C_0}{C_2}$
for some $C_0$.
By Theorem \ref{theorem:capset-interp-complete},
we can show that
$\encode{\<\set{},\rho,\rho^*\>}{C_0}{I_0}$ for some $I_0$.
We conclude this case by the IH
and the \ruleref{sc-trans} rule.

\emph{Case \rruleref{sc-set}}.
Then $C_1 = C_{11}\cup C_{12}$
$\subs{\G}{C_{11}}{C_2}$
and $\subs{\G}{C_{12}}{C_2}$.
We conclude by the IH and the \ruleref{sc-set} rule.

\emph{Case \rruleref{sc-var}}.
Then $C_1=\set{x}$
and $x:S\capt C_2\in \G$.
Then by Definition~\ref{def:proper-tcontext},
we can show that
$\subs{\Delta}{\rho(x)}{I_2}$
for some $\encode{\<D_2,\rho,\rho^*\>}{C_2}{I_2}$.
Since $\CAP\notin C_2$,
we invoke Theorem~\ref{theorem:interp-redundant} to show that
$\encode{\<\set{},\rho,\rho^*\>}{C_2}{I_2}$.
We can therefore conclude this case.

\emph{Case \rruleref{sc-elem}}.
Then $C_1\subseteq C_2$.
We can show that $D_1\subseteq D_2$,
and conclude this case by the \ruleref{sc-elem} rule.
\end{proof}

\begin{theorem}[Capture Set Translation is Monotonic (III)]
\label{theorem:capset-interp-monotonic-alt}
Given a proper translation context $\<D_1,\rho,\rho^*\>$
under contexts $\G$ and $\Delta$,
and $\subs{\G}{C_1}{C_2}$,
then for any $\encode{\<D_1,\rho,\rho^*\>}{C_1}{I_1}$,
there exists $D_2$ and $I_2$ such that
$\encode{\<D_2,\rho,\rho^*\>}{C_2}{I_2}$,
$\subs{\Delta}{I_1}{I_2}$,
and $D_2\subseteq I_1$.
\end{theorem}

\begin{proof}
By a case analysis on whether $\CAP\in C_2$.

\emph{Case $\CAP\in C_2$}.
Then set $D_2 = I_1$.
By Theorem~\ref{theorem:capset-interp-complete},
we can show that
$\encode{\<D_2,\rho,\rho^*\>}{C_2}{I_2}$ for some $I_2$.
One can verify that $D_2 = I_1\subseteq I_2$.
We can conclude this case by the \rruleref{sc-elem} rule.

\emph{Case $\CAP\notin C_2$}.
Then we can show that $\CAP\notin C_1$ as well.
By Theorem~\ref{theorem:interp-redundant}
we can show that
$\encode{\<\set{},\rho,\rho^*\>}{C_1}{I_1}$.
We then set $D_2=\set{}$ and
can show that
$\encode{\<D_2,\rho,\rho^*\>}{C_2}{I_2}$
for some $I_2$ by Theorem~\ref{theorem:capset-interp-complete}.
Then, we invoke Theorem~\ref{theorem:capset-interp-monotonic-helper} to show that
$\subs{\Delta}{I_1}{I_2}$.
This case can be concluded.
\end{proof}

\begin{theorem}[Type Translation is Complete]
\label{theorem:type-interp-complete}
Given any $\tau$ and $T$,
$\encode{\tau}{T}{U}$
for some $U$.
\end{theorem}

\begin{proof}[Proof Sketch]
By straightforward induction on the structure of $T$.
Use Theorem~\ref{theorem:capset-interp-complete} and the IH to conclude each case.
\end{proof}

\begin{theorem}[Capture Translation is Injective]
\label{theorem:capset-interp-injective}
Given $\encode{\tau}{C}{I_1}$
and $\encode{\tau}{C}{I_2}$,
we can show that $I_1 = I_2$.
\end{theorem}

\begin{proof}
By straightforward induction on the translation derivation.
\end{proof}

\begin{theorem}[Type Translation is Injective]
\label{theorem:type-interp-injective}
Given $\encode{\tau}{T}{U_1}$
and $\encode{\tau}{T}{U_2}$,
we can show that $U_1 = U_2$.
\end{theorem}

\begin{proof}[Proof Sketch]
By induction on the derivation of $\encode{D}{\G}{T}{U_1}$.
Make use of the IH and Theorem~\ref{theorem:capset-interp-injective} to conclude each case.
\end{proof}

\begin{definition}[Functional Notation for Translation]\label{defn:functional-notation}
Given that the capture set and type translation judgements are functional,
that is,
they are both injective and complete,
as shown in Theorem~\ref{theorem:capset-interp-injective},
Theorem~\ref{theorem:capset-interp-complete},
Theorem~\ref{theorem:type-interp-injective},
and Theorem~\ref{theorem:type-interp-complete},
we use $\enc{\cdot}^{\tau}$ to denote the output of the translation derivation
under the input $\tau$, and $\cdot$, 
where $\cdot$ can either be a capture set of a type.

Sometimes we write simply $\enc{\cdot}^{D}$
if the $\rho$ and $\rho^*$ in the translation context
are clear from the context.
\end{definition}

\begin{theorem}[Capture Set Replacement]
Given an already capturing type $T$,
$T\capt C$ replaces the capture set of $T$ with $C$.
Specifically, it is defined as:
\begin{equation*}
(S\capt C_0)\capt C \coloneqq S\capt C
\end{equation*}
where $T = S\capt C_0$.
\end{theorem}

\begin{theorem}[Translation Preserves Subtyping]
\label{theorem:subtyp-interp}
Given a proper translation context $\<D,\rho,\rho^*\>$
under contexts $\G$ and $\Delta$,
the subtyping derivation $\subs{\G}{S_1}{S_2}$
implies that
given any well-typed answer $a$ in System \capless{}
$\typ{\embed{C_a}^{D_1}}{\Delta}{a}{\embed{S_1\capt C}^{D_1}}$
we have
$\typ{\embed{C_a}^{D_2}}{\Delta}{a'}{\embed{S_2\capt C}^{D_2}}$
for some $D_2$.
\end{theorem}

\begin{proof}
By induction on the subtyping derivation.

\emph{Case \rruleref{top}}.
Then $S_2 = \top$.
We set $D_2 = D_1$ and $a' = a$.
We have $\embed{S_2\capt C}^{D_2} = \top\capt\embed{C}^{D_2}$
and conclude by the \rruleref{sub} and \rruleref{top} rules.

\emph{Case \rruleref{refl}}.
Then $S_1 = S_2$.
We set $D_2 = D_1$ and $a = a'$,
and conclude this case immediately.

\emph{Case \rruleref{trans}}.
Then $\subs{\G}{S_1}{S_0}$ and $\subs{\G}{S_0}{S_2}$
for some $S_0$.
We conclude by repeated application of the IHs.

\emph{Case \rruleref{boxed}}.
Then $S_1 = \BOX R_1\capt C_1$ and $S_2 = \BOX R_2\capt C_2$
for some $T_1$ and $T_2$.
We have
$\embed{S_1}^{D_1} = \forall[X<:\top](\forall[X<:\top]R_1\capt C_1)\capt C_1$.
Given
$\typ{\embed{C_a}^{D_1}}{\Delta}{\lambda[X<:\top]\lambda[X<:\top]a}{\embed{S_1}^{D_1}}$,
we invert the typing derivation to show that
$\typ{\embed{C_i}^{D_1}}{(\Delta,X<:\top,X<:\top)}{a}{\embed{R_1\capt C_2}^{D_1}}$
for some $C_i$.
Then, we invoke the IH to show that
$\typ{\embed{C_i}^{D_2}}{\Delta}{a_0}{\embed{R_2\capt C_2}^{D_2}}$
for some $a_0$ and $D_2$.
We set $a' = \lambda[X<:\top]\lambda[X<:\top]a_0$
and conclude this case by the \rruleref{tabs} rule.

\emph{Case \rruleref{fun}}.
Then $S_1 = \forall^{\alpha_1}(z:R_1\capt C_1)U_1$ and $S_2 = \forall^{\alpha_2}(z:R_2\capt C_2)U_2$.
In this case, we have
$\embed{S_1}^{D_1} = \forall[c_z<:\embed{C_1}]\forall[c_{z^*}](\forall(z:\embed{R_1}^{c_{z^*}}\capt\set{c_z})\EXCAP{c}\embed{U_1}^{c})$.
Given
$\typ{\embed{C_a}^{D_1}}{\Delta}{a}{\embed{S_1\capt C_f}^{D_1}}$,
we set $a'$ to
\begin{align*}
&\lambda[c_z<:\embed{C_2}]\lambda[c_{z^*}]\lambda(z:\embed{R_2}^{c_{z^*}}\capt\set{c_z})\\
&\quad\LET z_f = a\IN\\
&\quad\LET z_a = a_0\IN\\
&\quad\LET z_1 = z_f[\set{c_z}]\IN\\
&\quad\LET z_2 = z_1[D_0]\IN\\
&\quad\LET \<c_3,z_3\> = z_2\,z_a\IN\\
&\quad\LET z_o = a'_0\IN\\
&\quad\quad \<D'_0,z_o\>\\
\end{align*}
Here, $a_0$ and $D_0$ are the result of invoking the IH of
the subtyping derivation between $R_2$ and $R_1$
with $z$ as the input.
$a'_0$ and $D'_0$ is the result of invoking the IH of the subtyping derivation
between $U_1$ and $U_2$ with $z_3$ and $\set{c_3}$ as the input.
Note that from Theorem~\ref{theorem:capset-interp-monotonic-alt}
we can show that
$\subs{\Delta}{\embed{C_2}}{\embed{C_1}}$.
By $\alpha_1\preceq\alpha_2$
and definition of the type translation,
we can always find a $C_f'$ such that
$\subs{\Delta}{\embed{C_f}^{D_1}}{\embed{C_f}^{D_2}}$.
We therefore conclude this case by repeated application of
the \ruleref{cabs}, \ruleref{abs} and \ruleref{sub} rules.

\emph{Case \rruleref{tfun} and \rruleref{cfun}}.
Analogous to the \rruleref{box} case.

\emph{Case \rruleref{applied-p}}.
Then $S_1 = \kappa[T_1,\cdots,T_i,\cdots,T_n]$, $S_2 = \kappa[T_1,\cdots,T_i',\cdots,T_n]$,
$\kappa = (X_1^{\nu_1},\cdots,X_i^+,\cdots,X_n^{\nu_n})\mapsto S\in\cctx$,
and $\subs{\G}{T_i}{T_i'}$.
By the IH, we can show that
given any answer $a$ typed at $\embed{T_i}^{D_1}$
there exists a $D_2$ and an answer $a'$ such that
$a'$ can be typed at $\embed{T_i'}^{D_2}$.
By induction on the lexicographic order of
the size of $\cctx$ and the structure of $S$.
In each case, we can construct a way adapt the term.
Note that in the case of $S$ being an applied type,
we continue the induction by decreasing the size of $\cctx$,
since by the well-formedness of $\cctx$,
any type definitions only depend on ones coming before it.

\emph{Case \rruleref{applied-m}}.
Analogous to the \rruleref{applied-p} case.

\emph{Case \rruleref{dealias}}.
There are two subcases of two different directions of the subtyping derivation.
Let us consider the first case,
where $S_1 = \kappa[T_1,\cdots,T_n]$
$S_2 = [X_1:=T_1,\cdots,X_n:=T_n]S$,
$\kappa = (X_1^{\nu_1},\cdots,X_n^{\nu_n})\mapsto S\in\cctx$,
and $\forall X_i^+,\CAP\notin\dcs{\G}{T_i}$.
Note that for any choice of $D$, we have
$\embed{S_1}^{D} = \embed{[X_1:=T_1',\cdots,X_n:=T_n']S}^{D}$
where $\forall X_i^+,T'_i = \embed{T_i}^{D}$
and $\forall X_j^-,T'_j = T_j$.
For any choice of $D$, we have
$\embed{S_2}^{D} = \embed{[X_1:=T_1,\cdots,X_n:=T_n]S}^{D}$.
Note that since $\forall X_i^+,\CAP\notin\dcs{\G}{T_i}$,
the value of $T_i'$ is irrelavant to the choice of $D$,
and we can therefore show that
$\embed{S_1}^{D} = \embed{S_2}^{D}$.
This case is concluded by setting $D_2 = D_1$ and $a' = a$ given any $a$ that can be typed at $\embed{S_1}^{D_1}$.
The other subcase is analogous.
\end{proof}

\subsection{The Main Theorem}

\thmtranslation*

\begin{proof}[Proof of Theorem~\ref{theorem:typing-interp}]
By induction on the typing derivation $\typ{C}{\G}{t}{T}$.

\emph{Case \rruleref{var}}.
Then the typing derivation is of the form
$\typ{\set{x}}{\G}{x}{S'\capt\set{x}}$
where $\RefineReach{\set{x^*}}{S}{S'}$.
By Definition~\ref{def:proper-tcontext}, we know that
$x:\embed{S}^{\rho^*(x)}\capt\rho(x)\in\Delta$.
We can show that
$\embed{S'}^{\rho^*(x)} = \embed{S}^{\rho^*(x)}$.
Since $x$ is an answer, we need to choose a $t'$ that is an answer.
We set $D' = \rho^*(x)$ and $t' = x$ and conclude this case.

\emph{Case \rruleref{sub}}.
Then $\typ{C_0}{\G}{a}{T_0}$, $\subs{\G}{C_0}{C}$
and $\subs{\G}{T_0}{S}$.
By the IH, we can show that
$\typ{\embed{C_0}^{D_0}}{\Delta}{a'_0}{\embed{T_0}^{D_0}}$
for some capture set $D_0$ and answer $a'_0$.
We then conclude this case by Theorem~\ref{theorem:subtyp-interp} and Theorem~\ref{theorem:capset-interp-monotonic}.

\emph{Case \rruleref{box}}.
Then the typing derivation is of the form
$\typ{\set{}}{\G}{\BOX x_0}{\BOX (S_0\capt C_0)}$
where we have
$\typ{C_0}{\G}{x_0}{S_0\capt C_0}$.
By the IH, we can show that
$\typ{\embed{C_0}^{D_0}}{\Delta}{a'_0}{\embed{S_0\capt C_0}^{D_0}}$
for some $D_0$ and $a'_0$.
We let $D' = D_0$ and construct the following $t'$:
\begin{align*}
\lambda[X<:\top]\lambda[X<:\top]a'_0,
\end{align*}
which is also an answer.
We conclude this case by repeated application of the \ruleref{tabs} rule.

\emph{Case \rruleref{unbox}}.
Then the typing derivation is of the form
$\typ{C_0}{\G}{C\UNBOX x_0}{S_0\capt C_0}$
and we have
$\typ{C_0}{\G}{x_0}{\BOX S_0\capt C_0}$.
By analyzing the typing derivation, we can show that
$\typ{\set{}}{\G}{x_0}{\BOX S_0\capt C_0}$.
By the IH, we have
$\typ{\set{}}{\G}{a'_0}{\forall[X<:\top](\forall[X<:\top]S'_0\capt C'_0)\capt C'_0}$
where $C'_0 = \embed{C_0}^{D_0}$ and $S'_0 = \embed{S_0}^{D_0}$ for some $D_0$.
We set $D' = D_0$ and construct the following $t'$:
\begin{align*}
&\LET z_0 = a'_0\IN\\
&\LET z_1 = z_0[\top]\IN\\
&\quad z_1[\top]
\end{align*}
We can show that
$\typ{C'_0}{\Delta}{t'}{S'_0\capt C'_0}$
using the \ruleref{let} and \ruleref{tapp} rules,
and therefore conclude this case.

\emph{Case \rruleref{abs}}.
Then $t = \lambda^\alpha(x:R_0\capt C_0)t_0$, $T = (\forall^\alpha(x:U_0)T_0)\capt C_f$
and $\typ{C_f}{(\G,x:U_0)}{t_0}{T_0}$.
Let $\Delta' = (\Delta,c_x<:\embed{R_0},c_{x^*}<:\CAPK)$.
By the IH, we can show that one of the following holds:
\begin{enumerate}[(i)]
\item $\typ{\embed{C_f}^{D_0}}{\Delta'}{t'_0}{\embed{T_0}^{D_0}}$ for some $D_0$ and $t'_0$,
\item $\typ{\embed{C_f}^{D_0}}{\Delta'}{t'_0}{\EXCAP{c}\embed{T_0}^{\set{c}}}$ for some $D_0$ and $t'_0$.
\end{enumerate}
If the first case holds, we can construct the following $t'$:
\begin{align*}
&\lambda[c_x<:\embed{R_0}]\lambda[c_{x^*}]\lambda(x:\embed{R_0}^{c_{x^*}}\capt\set{c_x})\\
&\quad\LET z_r = t'_0\IN\\
&\quad\quad\<D_0,z_r\>\\
\end{align*}
Otherwise,
we can construct the following $t'$:
\begin{align*}
&\lambda[c_x<:\embed{R_0}]\lambda[c_{x^*}]\lambda(x:\embed{R_0}^{c_{x^*}}\capt\set{c_x})\\
&\quad t'_0\\
\end{align*}
In both cases, $t'$ is an answer.
We conclude by applying the \ruleref{cabs}, \ruleref{abs} and \ruleref{let} rules.

\emph{Case \rruleref{app}}.
Then $t = x\,y$,
$\typ{C}{\G}{x}{(\forall^\alpha(z:U_1)U_2)\capt C_f}$,
$y: S_y\capt D_y\in \G$,
$\subs{\G}{S_y\capt\set{y}}{U_1}$,
and $T = [z^*:=_{+}\dcs{\G}{S_y}][z:=y]U_2$.
We first analyze the typing derivation of $x$
and show that
$\typ{\set{x}}{\G}{x}{(\forall^\alpha(z:U_1)U_2)\capt\set{x}}$,
By the IH, we can show that
$\typ{\embed{\set{x}}^{D_0}}{\Delta}{a_x}{\embed{(\forall^\alpha(z:U_1)U_2)\capt\set{x}}^{D_0}}$
for some $D_0$ and $a_x$.
By Theorem~\ref{theorem:subtyp-interp},
we can show that
$\typ{\embed{\set{y}}^{D_1}}{\Delta}{a_y}{\embed{U_1}^{D_1}}$ for some $D_1$.
We construct the following $t'$:
\begin{align*}
&\LET z_x = a_x\IN\\
&\LET z_y = a_y\IN\\
&\LET z_1 = z_x[\rho(y)]\IN\\
&\LET z_2 = z_1[D_1]\IN\\
&\quad z_2\,z_y\\
\end{align*}
Then we can conclude this case by the \ruleref{app}, \ruleref{capp} and \ruleref{let} rules.
Note that, if $\alpha = \USE$,
the use set when typing $t'$ will include $D_1$.

\emph{Case \rruleref{tabs} and \rruleref{cabs}}.
Analogous to the \rruleref{abs} case.

\emph{Case \rruleref{tapp} and \rruleref{capp}}.
Analogous to the \rruleref{app} case.

\emph{Case \rruleref{let}}.
Then $t = \LET z = s_1\IN s_2$,
$\typ{C}{\G}{s_1}{T_1}$,
$\typ{C}{(\G,x:T_1)}{s_2}{T_2}$.
By applying the IH on the first typing derivation,
we can show that one of the following holds:
\begin{enumerate}[(i)]
\item $\typ{\embed{C}^{D_1}}{\Delta}{t'_1}{\embed{T_1}^{D_1}}$ for some $D_1$ and $t'_1$,
\item $\typ{\embed{C}^{D_1}}{\Delta}{t'_1}{\EXCAP{c}\embed{T_1}^{\set{c}}}$ for some $D_1$ and $t'_1$.
\end{enumerate}
If the first case holds,
we invoke the IH on the second typing derivation to show that
$\typ{\embed{C}^{D_2}}{\Delta,z:\embed{T_1}^{D_1}}{t'_2}{\embed{T_2}^{D_2}}$ for some $D_2$ and $t'_2$.
We construct the following $t'$:
\begin{align*}
&\LET z_1 = t'_1\IN\\
&\quad t'_2\\
\end{align*}
and conclude this case by the \ruleref{let} rule.
If the second case holds,
we invoke the IH on the second typing derivation to show that
$\typ{\embed{C}^{D_2}}{\Delta,c_z:\CAPK,z:\embed{T_1}^{c_z}}{t'_2}{\embed{T_2}^{D_2}}$ for some $D_2$ and $t'_2$,
and construct the following $t'$:
\begin{align*}
&\LET \<c_z,z_1\> = t'_1\IN\\
&\quad t'_2\\
\end{align*}
Then, this case can be concluded by the \ruleref{let-e} rule.

\end{proof}
 
\end{document}